\documentclass[11pt]{article}
\pdfoutput=1
\usepackage[T1]{fontenc}
\usepackage{lmodern}
\usepackage{slantsc}
\usepackage[protrusion=true,expansion=true]{microtype}
\usepackage{amsmath,amssymb,amsfonts,amsthm}
\usepackage{subcaption}
\usepackage{graphicx}
\usepackage{fullpage}
\usepackage{setspace}
\usepackage{mathpazo}
\usepackage[backref=page]{hyperref}
\usepackage{color}
\usepackage{wrapfig}
\usepackage{tikz}
\usetikzlibrary{decorations.pathreplacing}
\usepackage{algorithm}
\usepackage[noend]{algpseudocode}
\usepackage[framemethod=tikz]{mdframed}
\usepackage{xspace}
\usepackage{pgfplots}
\usepackage{framed}
\usepackage{thmtools}
\usepackage{thm-restate}
\usepackage{tabu}
\usepackage{fancyhdr}
\usepackage{ulem}
\pgfplotsset{compat=1.5}

\newtheorem{theorem}{Theorem}[section]
\newtheorem{corollary}[theorem]{Corollary}
\newtheorem{lemma}[theorem]{Lemma}

\newtheorem{definition}[theorem]{Definition}

\newtheorem{claim}[theorem]{Claim}

\newtheorem{fact}[theorem]{Fact}

\newtheorem{hypothesis}[theorem]{Hypothesis}

\newenvironment{proofof}[1]{\begin{trivlist} \item {\bf Proof
#1:~~}}
  {\qed\end{trivlist}}

\newcommand{\namedref}[2]{\hyperref[#2]{#1~\ref*{#2}}}
\newcommand{\thmlab}[1]{\label{thm:#1}}
\newcommand{\thmref}[1]{\namedref{Theorem}{thm:#1}}
\newcommand{\lemlab}[1]{\label{lem:#1}}
\newcommand{\lemref}[1]{\namedref{Lemma}{lem:#1}}

\newcommand{\corlab}[1]{\label{cor:#1}}
\newcommand{\corref}[1]{\namedref{Corollary}{cor:#1}}
\newcommand{\seclab}[1]{\label{sec:#1}}
\newcommand{\secref}[1]{\namedref{Section}{sec:#1}}

\newcommand{\factlab}[1]{\label{fact:#1}}
\newcommand{\factref}[1]{\namedref{Fact}{fact:#1}}

\newcommand{\figlab}[1]{\label{fig:#1}}
\newcommand{\figref}[1]{\namedref{Figure}{fig:#1}}
\newcommand{\alglab}[1]{\label{alg:#1}}
\renewcommand{\algref}[1]{\namedref{Algorithm}{alg:#1}}

\newcommand{\deflab}[1]{\label{def:#1}}
\newcommand{\defref}[1]{\namedref{Definition}{def:#1}}

\newcommand{\questref}[1]{\namedref{Question}{quest:#1}}
\newcommand{\questlab}[1]{\label{quest:#1}}
\newtheorem{question}[theorem]{Question}
\newcommand{\hypref}[1]{\namedref{Hypothesis}{hyp:#1}}
\newcommand{\hyplab}[1]{\label{hyp:#1}}

\def \LearnedCenters    {\mdef{\textsc{LearnedCenters}}}
\def \OPT    {\mdef{\mathsf{OPT}}}

\renewcommand{\O}[1]{\ensuremath{\mathcal{O}\left(#1\right)}}
\newcommand{\tO}[1]{\ensuremath{\tilde{\mathcal{O}}\left(#1\right)}}
\newcommand{\eps}{\varepsilon}
\newcommand{\con}{\mathsf{cov}}
\newcommand{\JCH}{$\mathsf{JCH}$\xspace}
\newcommand{\JCHD}{$\mathsf{JCH}^*$\xspace}
\newcommand{\JCHDB}{$\mathsf{Balanced-JCH}^*$\xspace}
\newcommand{\jc}{Johnson Coverage\xspace}
\newcommand{\jcd}{Johnson Coverage$^*$\xspace}
\newcommand{\jcdb}{Balanced Johnson Coverage$^*$\xspace}
\newcommand{\C}{\mathcal{C}}

\newcommand{\mkc}{Max $k$-Coverage\xspace}

\newcommand{\NP}{$\mathsf{NP}$\xspace}

\newcommand{\calm}{\mathcal{M}}
\newcommand{\calp}{\mathcal{P}}
\newcommand{\calh}{\mathcal{H}}
\newcommand{\cald}{\mathcal{D}}
\newcommand{\NN}{\mathbb{N}}

\def \calF    {\mdef{\mathcal{F}}}

\def \calS    {\mdef{\mathcal{S}}}

\newcommand{\mdef}[1]{{\ensuremath{#1}}\xspace}  

\DeclareMathOperator*{\argmin}{argmin}

\DeclareMathOperator*{\polylog}{polylog}
\DeclareMathOperator*{\poly}{poly}

\newcommand{\flr}[1]{\mdef{\left\lfloor#1\right\rfloor}}              
\newcommand{\ceil}[1]{\mdef{\left\lceil#1\right\rceil}}               
\newcommand{\E}[2][]{\mdef{\underset{#1}{\mathbb{E}}\left[#2\right]}} 

\newcommand{\ignore}[1]{}

\newif\ifnotes\notestrue 
\ifnotes
\newcommand{\samson}[1]{\textcolor{purple}{{\bf (Samson:} {#1}{\bf ) }} \marginpar{\tiny\bf
             \begin{minipage}[t]{0.5in}
               \raggedright S:
            \end{minipage}}}

\newcommand{\david}[1]{\textcolor{purple}{{\bf (David:} {#1}{\bf ) }} \marginpar{\tiny\bf
             \begin{minipage}[t]{0.5in}
               \raggedright D:
            \end{minipage}}} 
\else
\newcommand{\samson}[1]{}
\newcommand{\david}[1]{}
\fi

\makeatletter
\renewcommand*{\@fnsymbol}[1]{\textcolor{mahogany}{\ensuremath{\ifcase#1\or *\or \dagger\or \ddagger\or
 \mathsection\or \triangledown\or \mathparagraph\or \|\or **\or \dagger\dagger
   \or \ddagger\ddagger \else\@ctrerr\fi}}}
\makeatother

\providecommand{\email}[1]{\href{mailto:#1}{\nolinkurl{#1}\xspace}}

\definecolor{mahogany}{rgb}{0.75, 0.25, 0.0}
\definecolor{darkblue}{rgb}{0.0, 0.0, 0.55}
\definecolor{darkpastelgreen}{rgb}{0.01, 0.75, 0.24}
\definecolor{darkgreen}{rgb}{0.0, 0.2, 0.13}
\definecolor{darkgoldenrod}{rgb}{0.72, 0.53, 0.04}
\definecolor{darkred}{rgb}{0.55, 0.0, 0.0}
\definecolor{forestgreenweb}{rgb}{0.13, 0.55, 0.13}
\definecolor{greencss}{rgb}{0.0, 0.5, 0.0}
\definecolor{bleudefrance}{rgb}{0.19, 0.55, 0.91}

\hypersetup{
     colorlinks   = true,
     citecolor    = mahogany,
	 linkcolor	  = darkpastelgreen
}

\fancypagestyle{pg}
{
\lhead{}
\rhead{}
\cfoot{--\ \thepage\ --}

}

\begin{document}

\allowdisplaybreaks

\title{On Approximability of $\ell_2^2$ Min-Sum Clustering}

\author{
\hspace{0.4in}
Karthik C. S.\thanks{Rutgers University.
E-mail: \email{karthik.cs@rutgers.edu}}
\hspace{0.4in}
\and 
Euiwoong Lee\thanks{University of Michigan. 
E-mail: \email{euiwoong@umich.edu}}
\hspace{0.4in}
\and 
Yuval Rabani\thanks{The Hebrew University of Jerusalem. 
E-mail: \email{yrabani@cs.huji.ac.il}}
\hspace{0.4in}
\vspace{-0.2in}
\and 
Chris Schwiegelshohn
\thanks{Aarhus University. 
E-mail: \email{cschwiegelshohn@gmail.com}} 
\and 
Samson Zhou\thanks{Texas A\&M University. 
E-mail: \email{samsonzhou@gmail.com}}
}
\date{}

\maketitle

\begin{abstract}
The $\ell_2^2$ min-sum $k$-clustering problem is to partition an input set into clusters $C_1,\ldots,C_k$ to minimize $\sum_{i=1}^k\sum_{p,q\in C_i}\|p-q\|_2^2$. The objective is a density-based clustering and can be more effective than the traditional centroid-based clustering like $k$-median and $k$-means in capturing complex structures in data that may not be linearly separable, such as when the clusters have irregular, non-convex shapes or are overlapping. Although $\ell_2^2$ min-sum $k$-clustering is NP-hard, it is not known whether it is NP-hard to approximate $\ell_2^2$ min-sum $k$-clustering beyond a certain factor. 

In this paper, we give the first hardness-of-approximation result for the $\ell_2^2$ min-sum $k$-clustering problem. We show that it is NP-hard to approximate the objective to a factor better than $1.056$ and moreover, assuming a balanced variant of the Johnson Coverage Hypothesis, it is  NP-hard to approximate the objective to a factor better than $1.327$. 

We then complement our hardness result by giving a nearly linear time parameterized PTAS for $\ell_2^2$ min-sum $k$-clustering running in time $O\left(n^{1+o(1)}d\cdot \exp((k\cdot\varepsilon^{-1})^{O(1)})\right)$, where $d$ is the underlying dimension of the input dataset.

Finally, we consider a learning-augmented setting, where the algorithm has access to an oracle that outputs a label $i\in[k]$ for input point, thereby implicitly partitioning the input dataset into $k$ clusters that induce an approximately optimal solution, up to some amount of adversarial error $\alpha\in\left[0,\frac{1}{2}\right)$. We give a polynomial-time algorithm that outputs a $\frac{1+\gamma\alpha}{(1-\alpha)^2}$-approximation to $\ell_2^2$ min-sum $k$-clustering, for a fixed constant $\gamma>0$. Therefore, our algorithm improves smoothly with the performance of the oracle and can be used to achieve approximation guarantees better than the NP-hard barriers for sufficiently accurate oracles. 
\end{abstract}



\thispagestyle{empty}
\newpage
\setcounter{page}{1}

\section{Introduction}
Clustering is a fundamental technique that partitions an input dataset into distinct groups called clusters, which facilitate the identification and subsequent utilization of latent structural properties underlying the dataset. 
Consequently, various formulations of clustering are used across a wide range of applications, such as computational biology, computer vision, data mining, and machine learning~\cite{jain1999data,xu2005survey}. 

Ideally, the elements of each cluster are more similar to each other than to elements in other clusters. 
To formally capture this notion, a dissimilarity metric is often defined on the set of input elements, so that more closer objects in the metric correspond to more similar objects. 
Perhaps the most natural goal would be to minimize the intra-cluster dissimilarity in a partitioning of the input dataset. 
This objective is called the \emph{min-sum $k$-clustering} problem and has received significant attention due to its intuitive clustering objective~\cite{Guttmann-BeckH98a,Indyk99,Matousek00,Schulman00,BartalCR01,VegaKKR03,CzumajS07,AloiseDHP09,BehsazFSS19,BanerjeeOR21,Cohen-AddadSL21}. 

In this paper, we largely focus on the $\ell_2^2$ min-sum $k$-clustering formulation. 
Formally, the input is a set $X$ of $n$ points in $\mathbb{R}^d$ and the goal is to partition $X=C_1\dot\cup\cdots\dot\cup C_k$ into $k$ clusters to minimize the quantity
\[\min_{C_1,\ldots,C_k}\sum_{i=1}^k\sum_{p,q\in C_i}\|p-q\|_2^2,\]
where $\|\cdot\|_2$ denotes the standard Euclidean $\ell_2$ norm. 

Whereas classical centroid-based clustering problems such as $k$-means and $k$-median leverage distances between data points and cluster centroids to identify convex shapes that partition the dataset, min-sum $k$-clustering is a density-based clustering that can handle complex structures in data that may not be linearly separable. 
In particular, min-sum $k$-clustering can be more effective than traditional centroid-based clustering in scenarios where clusters have irregular, non-convex shapes or overlapping clusters.
A simple example of the ability of min-sum clustering to capture more natural structure is an input that consists of two concentric dense rings of points in the plane. 
Whereas min-sum clustering can partition the points into the separate rings, centroid-based clustering will instead create a separating hyperplane between these points, thereby ``incorrectly'' grouping together points of different rings. 
See \figref{fig:cluster} for an example of the ability of min-sum clustering to capture natural structure in cases where centroid-based clustering fails. 

Moreover, min-sum clustering satisfies Kleinberg’s consistency axiom~\cite{Kleinberg02}, which informally demands that the optimal clustering for a particular objective should be preserved when distances between points inside a cluster are shrunk and distances between points in different clusters are expanded. 
By contrast, many centroid-based clustering objectives, including $k$-means and $k$-median, do not satisfy Kleinberg's consistency axiom~\cite{MahajanNV12}. 

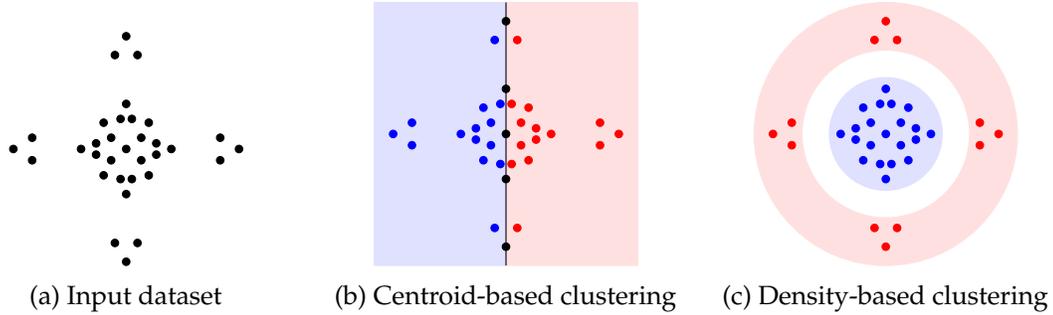
\begin{figure*}[!htb]
\centering
\begin{subfigure}[b]{0.3\textwidth}
\centering
\begin{tikzpicture}[scale=0.5]
\filldraw (0,0) circle (0.1);
\filldraw (1.2,0) circle (0.1);
\filldraw (-1.2,0) circle (0.1);
\filldraw (0,1.2) circle (0.1);
\filldraw (0,-1.2) circle (0.1);
\filldraw (0.6,0.7) circle (0.1);
\filldraw (-0.6,-0.7) circle (0.1);
\filldraw (0.6,-0.7) circle (0.1);
\filldraw (-0.6,0.7) circle (0.1);
\filldraw (0.4,0.3) circle (0.1);
\filldraw (0.4,-0.3) circle (0.1);
\filldraw (-0.4,-0.3) circle (0.1);
\filldraw (-0.4,0.3) circle (0.1);
\filldraw (0.8,0.15) circle (0.1);
\filldraw (0.8,-0.15) circle (0.1);
\filldraw (-0.8,0.15) circle (0.1);
\filldraw (-0.8,-0.15) circle (0.1);
\filldraw (0.15,0.8) circle (0.1);
\filldraw (0.15,-0.8) circle (0.1);
\filldraw (-0.15,0.8) circle (0.1);
\filldraw (-0.15,-0.8) circle (0.1);

\filldraw (2.5,0.3) circle (0.1);
\filldraw (-2.5,0.3) circle (0.1);
\filldraw (2.5,-0.3) circle (0.1);
\filldraw (-2.5,-0.3) circle (0.1);
\filldraw (3,0) circle (0.1);
\filldraw (-3,0) circle (0.1);
\filldraw (0.3,2.5) circle (0.1);
\filldraw (0.3,-2.5) circle (0.1);
\filldraw (-0.3,2.5) circle (0.1);
\filldraw (-0.3,-2.5) circle (0.1);
\filldraw (0,3) circle (0.1);
\filldraw (0,-3) circle (0.1);
\end{tikzpicture}
\caption{Input dataset}
\figlab{fig:cluster:a}
\end{subfigure}
\begin{subfigure}[b]{0.3\textwidth}
\centering
\begin{tikzpicture}[scale=0.5]
\filldraw[red!30!, opacity=0.4] (0,-3.5) rectangle +(3.5,7);
\filldraw[blue!30!, opacity=0.4] (0,-3.5) rectangle +(-3.5,7);

\filldraw (0,0) circle (0.1);
\filldraw[red] (1.2,0) circle (0.1);
\filldraw[blue] (-1.2,0) circle (0.1);
\filldraw (0,1.2) circle (0.1);
\filldraw (0,-1.2) circle (0.1);
\filldraw[red] (0.6,0.7) circle (0.1);
\filldraw[blue] (-0.6,-0.7) circle (0.1);
\filldraw[red] (0.6,-0.7) circle (0.1);
\filldraw[blue] (-0.6,0.7) circle (0.1);
\filldraw[red] (0.4,0.3) circle (0.1);
\filldraw[red] (0.4,-0.3) circle (0.1);
\filldraw[blue] (-0.4,-0.3) circle (0.1);
\filldraw[blue] (-0.4,0.3) circle (0.1);
\filldraw[red] (0.8,0.15) circle (0.1);
\filldraw[red] (0.8,-0.15) circle (0.1);
\filldraw[blue] (-0.8,0.15) circle (0.1);
\filldraw[blue] (-0.8,-0.15) circle (0.1);
\filldraw[red] (0.15,0.8) circle (0.1);
\filldraw[red] (0.15,-0.8) circle (0.1);
\filldraw[blue] (-0.15,0.8) circle (0.1);
\filldraw[blue] (-0.15,-0.8) circle (0.1);

\filldraw[red] (2.5,0.3) circle (0.1);
\filldraw[blue] (-2.5,0.3) circle (0.1);
\filldraw[red] (2.5,-0.3) circle (0.1);
\filldraw[blue] (-2.5,-0.3) circle (0.1);
\filldraw[red] (3,0) circle (0.1);
\filldraw[blue] (-3,0) circle (0.1);
\filldraw[red] (0.3,2.5) circle (0.1);
\filldraw[red] (0.3,-2.5) circle (0.1);
\filldraw[blue] (-0.3,2.5) circle (0.1);
\filldraw[blue] (-0.3,-2.5) circle (0.1);
\filldraw (0,3) circle (0.1);
\filldraw (0,-3) circle (0.1);

\draw (0,3.5) -- (0,-3.5);
\end{tikzpicture}
\caption{Centroid-based clustering}
\figlab{fig:cluster:b}
\end{subfigure}
\begin{subfigure}[b]{0.3\textwidth}
\centering
\begin{tikzpicture}[scale=0.5]
\filldraw[red!30!, opacity=0.4]  (0,0) circle (3.5);
\filldraw[white]  (0,0) circle (2.2);
\filldraw[blue!30!, opacity=0.4]  (0,0) circle (1.5);

\filldraw[blue] (0,0) circle (0.1);
\filldraw[blue] (1.2,0) circle (0.1);
\filldraw[blue] (-1.2,0) circle (0.1);
\filldraw[blue] (0,1.2) circle (0.1);
\filldraw[blue] (0,-1.2) circle (0.1);
\filldraw[blue] (0.6,0.7) circle (0.1);
\filldraw[blue] (-0.6,-0.7) circle (0.1);
\filldraw[blue] (0.6,-0.7) circle (0.1);
\filldraw[blue] (-0.6,0.7) circle (0.1);
\filldraw[blue] (0.4,0.3) circle (0.1);
\filldraw[blue] (0.4,-0.3) circle (0.1);
\filldraw[blue] (-0.4,-0.3) circle (0.1);
\filldraw[blue] (-0.4,0.3) circle (0.1);
\filldraw[blue] (0.8,0.15) circle (0.1);
\filldraw[blue] (0.8,-0.15) circle (0.1);
\filldraw[blue] (-0.8,0.15) circle (0.1);
\filldraw[blue] (-0.8,-0.15) circle (0.1);
\filldraw[blue] (0.15,0.8) circle (0.1);
\filldraw[blue] (0.15,-0.8) circle (0.1);
\filldraw[blue] (-0.15,0.8) circle (0.1);
\filldraw[blue] (-0.15,-0.8) circle (0.1);

\filldraw[red] (2.5,0.3) circle (0.1);
\filldraw[red] (-2.5,0.3) circle (0.1);
\filldraw[red] (2.5,-0.3) circle (0.1);
\filldraw[red] (-2.5,-0.3) circle (0.1);
\filldraw[red] (3,0) circle (0.1);
\filldraw[red] (-3,0) circle (0.1);
\filldraw[red] (0.3,2.5) circle (0.1);
\filldraw[red] (0.3,-2.5) circle (0.1);
\filldraw[red] (-0.3,2.5) circle (0.1);
\filldraw[red] (-0.3,-2.5) circle (0.1);
\filldraw[red] (0,3) circle (0.1);
\filldraw[red] (0,-3) circle (0.1);
\end{tikzpicture}
\caption{Density-based clustering}
\figlab{fig:cluster:c}
\end{subfigure}
\caption{Clustering of input dataset in \figref{fig:cluster:a} with $k=2$. 
\figref{fig:cluster:b} is an optimal centroid-based clustering, e.g., $k$-median or $k$-means, while the more natural clustering in \figref{fig:cluster:c} is an optimal density-based clustering, e.g., $\ell_2$ min-sum $k$-clustering.}
\figlab{fig:cluster}
\end{figure*}

On the other hand, theoretical understanding of density-based clustering objectives such as min-sum $k$-clustering is far less developed than that of their centroid-based counterparts. 
It can be shown that min-sum $k$-clustering with the $\ell_2^2$ cost function is NP-hard, using arguments from \cite{AloiseDHP09}. 
The problem is NP-hard even for $k=2$~\cite{VegaK01} in the metric case, where the only available information about the points is their pairwise dissimilarity, c.f., \secref{sec:related} for a summary of additional related work.  
In fact, for general $k$ in the metric case, it is NP-hard to approximate the problem within a $1.415$-multiplicative factor~\cite{GuruswamiI03,Cohen-AddadSL21}. 
However, no such hardness of approximation is known for the Euclidean case, i.e., $\ell_2^2$ min-sum, where the selected cost function is based on the geometry of the underlying space; the only known lower bound is the NP-hardness of the problem~\cite{AloiseDHP09,BanerjeeOR21,AKP24}. 
Thus a fundamental open question is:
\begin{question}
\questlab{quest:apx}
Is $\ell_2^2$ min-sum $k$-clustering APX-hard? 
That is, does there exist a natural hardness-of-approximation barrier for polynomial time algorithms?
\end{question}
Due to existing APX-hardness results for centroid-based clustering such as $k$-means and $k$-median~\cite{LeeSW17,Cohen-AddadS19,Cohen-AddadSL22}, it is widely believed that $\ell_2^2$ min-sum clustering is indeed APX-hard. 
Thus, there has been a line of work preemptively seeking to overcome such limitations. 
Indeed, on the positive side, \cite{InabaKI94} first showed that min-sum $k$-clustering in the $d$-dimensional $\ell_2^2$ case can be solved in polynomial time if both $d$ and $k$ are constants. 
For general graphs and fixed constant $k$, \cite{Guttmann-BeckH98a} gave a $2$-approximation algorithm using runtime $n^{\O{k}}$. 
The approximation guarantees were improved by a line of work~\cite{Indyk99,Matousek00,Schulman00}, culminating in polynomial-time approximation schemes by~\cite{VegaKKR03} for both the $\ell_2^2$ case and the metric case. 
Without any assumptions on $d$ and $k$, \cite{BartalCR01} introduced a polynomial algorithm that achieves an $\O{\frac{1}{\eps}\log^{1+\eps}n}$-multiplicative approximation. 
Therefore, a long-standing direction in the study of $\ell_2^2$ min-sum clustering is:
\begin{question}
\questlab{quest:algs}
How can we algorithmically bridge the gap between the NP-hardness of solving the $\ell_2^2$ min-sum clustering and the large multiplicative guarantees of existing approximation algorithms?
\end{question}

A standard approach to circumvent poor dependencies on the size of the input dataset is to sparsify the problem. Informally, we would like to reduce the search space by considering fewer candidate solutions and reduce the dependency on the number of input points by aggregating them.
For min-sum clustering this is a particular challenge, as a candidate solution is a partition and the cost of that partition depends on all pairwise distances between all the points.
While sparsification algorithms exist for graph clustering \cite{JambulapatiLS23,Lee23} and $k$-means clustering \cite{Cohen-AddadSS21,Cohen-AddadLSSS22}, where the output is typically called a coreset, similar constructions are not known to exist for min-sum clustering.

Another standard approach to overcome limitations inherent in worst-case impossibility barriers is to consider beyond worst case analysis. 
To that end, recent works have observed that in many applications, auxiliary information is often available and can potentially form the foundation upon which machine learning models are built. 
For example, previous datasets with potentially similar behavior can be used as training data for models to label future datasets. 
However, these heuristics lack provable guarantees and can produce embarrassingly inaccurate predictions when generalizing to unfamiliar inputs \cite{SzegedyZSBEGF13}. 
Nevertheless, \emph{learning-augmented algorithms}~\cite{MitzenmacherV20} have been shown to achieved both good algorithmic performance when the oracle is accurate, i.e., consistency, and standard algorithmic performance when the oracle is inaccurate, i.e., robustness for a wide range of settings, such as data structure design~\cite{KraskaBCDP18,Mitzenmacher18,LinLW22}, algorithms with faster runtime~\cite{DinitzILMV21,ChenSVZ22,DaviesMVW23}, online algorithms with better competitive ratio~\cite{PurohitSK18,GollapudiP19,LattanziLMV20,WangLW20,WeiZ20,BamasMS20,ImKQP21,LykourisV21,AamandCI22,Anand0KP22,AzarPT22,GrigorescuLSSZ22,KhodakBTV22,JiangLLTZ22,AntoniadisCEPS23,ShinLLA23}, and streaming algorithms that are more space-efficient~\cite{HsuIKV19,IndykVY19,JiangLLRW20,ChenIW22,ChenEILNRSWWZ22,LLLVW23}. 
In particular, \cite{ErgunFSWZ22,NguyenCN23} introduce algorithms for $k$-means and $k$-median clustering that can achieve approximation guarantees beyond the known APX-hardness limits. 

\subsection{Our Contributions}
In this paper, we perform a comprehensive study on the approximability of the $\ell_2^2$ min-sum $k$-clustering by answering \questref{quest:apx} and \questref{quest:algs}. 

\paragraph{Hardness-of-approximation of min-sum $k$-clustering.}
We first answer \questref{quest:apx} in the affirmative, by not only showing that the $\ell_2^2$ min-sum $k$-clustering is APX-hard but further giving an explicit constant NP-hardness of approximation result for the problem. 
\begin{theorem}[Hardness of approximation of $\ell_2^2$ min-sum $k$-clustering]
\thmlab{thm:mainlb}
It is NP-hard to approximate $\ell_2^2$ min-sum $k$-clustering to a factor better than $1.056$. 
Moreover, assuming the Dense and Balanced Johnson Coverage Hypothesis (\JCHDB), we have that the $\ell_2^2$ min-sum $k$-clustering is NP-hard to approximate to a factor better than $1.327$. 
\end{theorem}

We remark that \JCHDB in the theorem statement above is simply a balanced formulation of the recently introduced Johnson Coverage Hypothesis  \cite{Cohen-AddadSL22}. 

\textbf{Fast polynomial-time approximation scheme.}
In light of \thmref{thm:mainlb}, a natural question would be to closely examine alternative conditions in which we can achieve a $(1+\eps)$-approximation to min-sum $k$-clustering, i.e., \questref{quest:algs}. 
To that end, there are a number of existing polynomial-time approximation schemes (PTAS)~\cite{Indyk99,Matousek00,Schulman00,VegaKKR03}, the best of which uses runtime $n^{\O{k/\eps^2}}$ for the $\ell_2^2$ case. 
However, as noted by \cite{CzumajS07}, even algorithms with runtime quadratic in the size $n$ of the input dataset are generally not sufficiently scalable to handle large datasets. 
In this paper, we present an algorithm with a running time that is nearly nearly linear. Specifically, we show
\begin{restatable}{theorem}{thmptas}
\label{thm:ptas}
    There exists an algorithm running in time
    $$O\left(n^{1+o(1)}d\cdot  2^{\eta\cdot k^2 \cdot \varepsilon^{-12} \log^2(k/(\varepsilon\delta))}\right),$$ 
    for some absolute constant $\eta$, that computes a $(1+\varepsilon)$-approximate solution to $\ell_2^2$ $k$-MinSum Clustering with probability $1-\delta$.
\end{restatable}
We again emphasize that the runtime of \ref{thm:ptas} is linear in the size $n$ of the input dataset, though it has exponential dependencies in both the number $k$ of clusters and the approximation parameter $\eps>0$. 
By contrast, the best previous PTAS uses runtime $n^{\O{k/\eps^2}}$~\cite{CzumajS07}, which has substantially worse dependency on the size $n$ of the input dataset. 

\paragraph{Learning-augmented algorithms.}
Unfortunately, exponential dependencies on the number $k$ of clusters can still be prohibitive for moderate values of $k$. 
To that end, we turn our attention to learning-augmented papers. 
We consider the standard \emph{label oracle} model for clustering, where the algorithm has access to an oracle that provides a label for each input point. 
Formally, for each point $x$ of the $n$ input points, the oracle outputs a label $i\in[k]$ for $x$, so that the labels implicitly partition the input dataset into $k$ clusters that induce an approximately optimal solution. 
However, the oracle also has some amount of adversarial error that respects the precision and recall of each cluster; we defer the formal definition to \defref{def:label:oracle}. 

One of the reasons label oracles have been used for learning-augmented algorithms for clustering is their relative ease of acquisition via machine learning models that are trained on a similar distribution of data. 
For example, a smaller separate dataset can be observed and used as a ``training'' data, an input to some heuristic to cluster the initial data, which we can then use to form a predictor for the actual input dataset. 
Indeed, implementations of label oracles have been shown to perform well in practice~\cite{ErgunFSWZ22,NguyenCN23}.

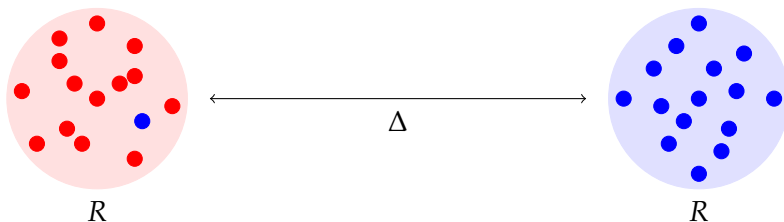
\begin{figure*}[!htb]
\centering
\begin{tikzpicture}[scale=1]
\filldraw[red!30!, opacity=0.4]  (-3,0) circle (1.2);
\filldraw[red] (-3,0) circle (0.1);
\filldraw[red] (-2.5,0.3) circle (0.1);
\filldraw[red] (-3.2,-0.6) circle (0.1);
\filldraw[red] (-3.4,-0.4) circle (0.1);
\filldraw[red] (-2.7,0.2) circle (0.1);
\filldraw[red] (-3,1) circle (0.1);
\filldraw[red] (-2,-0.1) circle (0.1);
\filldraw[red] (-4,0.1) circle (0.1);
\filldraw[red] (-3.5,0.5) circle (0.1);
\filldraw[red] (-2.5,0.7) circle (0.1);
\filldraw[red] (-3.8,-0.6) circle (0.1);
\filldraw[red] (-3.3,0.2) circle (0.1);
\filldraw[red] (-3.5,0.8) circle (0.1);
\filldraw[red] (-2.5,-0.8) circle (0.1);
\filldraw[blue] (-2.4,-0.3) circle (0.1);

\filldraw[blue!30!, opacity=0.4]  (5,0) circle (1.2);
\filldraw[blue] (5,0) circle (0.1);
\filldraw[blue] (6,0) circle (0.1);
\filldraw[blue] (4,0) circle (0.1);
\filldraw[blue] (5,1) circle (0.1);
\filldraw[blue] (5,-1) circle (0.1);
\filldraw[blue] (4.6,-0.6) circle (0.1);
\filldraw[blue] (5.4,-0.4) circle (0.1);
\filldraw[blue] (4.4,0.4) circle (0.1);
\filldraw[blue] (5.6,0.6) circle (0.1);
\filldraw[blue] (4.8,-0.3) circle (0.1);
\filldraw[blue] (5.2,0.4) circle (0.1);
\filldraw[blue] (5.3,-0.7) circle (0.1);
\filldraw[blue] (4.7,0.7) circle (0.1);
\filldraw[blue] (5.5,0.1) circle (0.1);
\filldraw[blue] (4.5,-0.1) circle (0.1);
\draw[<->] (-1.5,0) -- (3.5,0);

\node at (-3,-1.5){$R$};
\node at (5,-1.5){$R$};
\node at (1,-0.3){$\Delta$};
\end{tikzpicture}
\caption{Note that with arbitrarily small error rate, i.e., $\frac{1}{n}$, a single mislabeled point among the $n$ input points causes the resulting clustering to be arbitrarily bad for $\Delta\gg n^2\cdot R$.}
\figlab{fig:aug:bad}
\end{figure*}

We also remark that perhaps counter-intuitively, a label oracle with arbitrarily high accuracy does not trivialize the problem. 
In particular, the na\"{i}ve algorithm of outputting the clustering induced by the labels does not work. 
As a simple example, consider an input dataset where half of the $n$ points are at $x=0$ and the other half of the points are at $x=1$. 
Then for $k=2$, the clear optimal clustering is to cluster the points at the origin together, and cluster the points at $x=1$ together, which induces the optimal cost of zero. 
However, if even one of the $n$ points is incorrect, then the clustering output by the labels has cost at least $1$. 
Therefore, even with error rate as small as $\frac{1}{n}$, the multiplicative approximation of the na\"{i}ve algorithm can be arbitrarily bad. 
See \figref{fig:aug:bad} for an illustration of this example. 
Of course, this example does not rule out more complex algorithms that combines the labels with structural properties of optimal clustering and indeed, our algorithm utilizes such properties. 

We give a polynomial-time algorithm for the $\ell_2^2$ min-sum $k$-clustering that can provide guarantees beyond the computational limits of \thmref{thm:mainlb}, given a sufficiently accurate oracle. 
\begin{theorem}
\thmlab{thm:learn:aug}
There exists a polynomial-time algorithm that uses a label predictor with error rate $\alpha\in\left[0,\frac{1}{2}\right)$ and outputs a $\frac{1+\gamma\alpha}{(1-\alpha)^2}$-approximation to the $\ell_2^2$ min-sum $k$-clustering problem, where $\gamma=7.7$ for $\alpha\in\left[0,\frac{1}{7}\right)$ or $\gamma=\frac{5\alpha-2\alpha^2}{(1-2\alpha)(1-\alpha)}$ for $\alpha\in\left[0,\frac{1}{2}\right)$. 
\end{theorem}
We remark that \thmref{thm:learn:aug} does not require the true error rate $\alpha$ as an input parameter. 
Because we are in an offline setting, where can run \thmref{thm:learn:aug} multiple times with guesses for the true error rate $\alpha$, in decreasing powers of $\frac{1}{\lambda}$ for any constant $\lambda>1$. 
We can then compare the resulting clustering output by each guess for $\alpha$ and take the output the best clustering. 

\subsection{Technical Overview}

\paragraph{Hardness of approximation.}
Recently, the authors of \cite{Cohen-AddadSL22} put forth the Johnson Coverage Hypothesis (\JCH) and introduced a framework to obtain (optimal) hardness of approximation results for $k$-median and $k$-means in $\ell_p$-metrics. 
The proof of \thmref{thm:mainlb} builds on this framework.

\JCH roughly asserts that for large enough constant $z$,  given as input an integer $k$ and a collection of $z$-sets (i.e., sets each of size $z$) over some  universe, it is NP-hard to distinguish the completeness case where  there is a collection $C$ of $k$ many $(z-1)$-sets such that every input set is covered\footnote{A $(z-1)$-set covers a $z$-set if the former is a subset of the latter.} by some set in $C$, from the soundness case where every collection $C$ of $k$ many $(z-1)$-sets does not cover much more than $1-\frac{1}{e}$ fraction of the input sets (see \hypref{hyp:jch} for a formal statement). 

In this paper, we consider a natural generalization of \JCH, called \JCHDB, where we assume that the number of input sets is ``dense'', i.e., $\omega(k)$, and more importantly that in the completeness case,  the collection $C$ covers the input $z$-sets in a balanced manner, i.e., we can partition the input to $k$ equal parts such that each part is completely covered by a single set in $C$ (see \hypref{hyp:jchdb} for a formal statement).

We now sketch the proof of \thmref{thm:mainlb} assuming \JCHDB. Given a collection of $m$ many $z$-sets over a universe $[n]$ as input, we create a point for each input set, which is simply the characteristic vector of the set as a subset of $[n]$, i.e., the points are all $n$-dimensional Boolean vectors of Hamming
weight $z$. 

In the
completeness case, from the guarantees of \JCHDB, it is easy to see that the points created can be divided into $k$ equal clusters of size $m/k$ such that all the $z$-sets of a cluster are completely covered by a single $(z-1)$-set. This implies that the squared Euclidean distance between a pair of points within a cluster is exactly $2$ and thus the $\ell_2^2$ min-sum $k$-clustering cost is $k\cdot 2\cdot (m/k)(m/k-1)\approx 2m^2/k$.

On the other hand, in the soundness case, we first use the density guarantees of \JCHDB to argue that most clusters are not small. Then suppose that we had a low cost $\ell_2^2$ min-sum $k$-clustering, we look at a typical cluster and observe that the squared distance  of any two points  in the cluster must be a positive even integer, and it is exactly 2 only when the two input sets corresponding to the points intersect on a $(z-1)$-set. Thus, if the cost of the clustering is close to $\alpha \cdot 2m^2/k$ (for some $\alpha\ge 1$), then we argue (using convexity) that for a typical cluster that there must be a $(z-1)$-set that covers $(1-\alpha')m/k$ many $z$-sets in that cluster, where $\alpha'$ depends on $\alpha$. Thus, from this we decode $k$-many $(z-1)$-sets which cover a large fraction of the input $z$-sets. 

In order to obtain the unconditional NP-hardness result, much like in \cite{Cohen-AddadSL22}, we need to extend the above reduction to a more general problem. This is indeed established in \thmref{thm:mainreduction}, and after this we  prove a special case of a generalization of \JCHDB (when $z=3$) which is done  in \thmref{thm:minsumNP}  and this involved proving additional properties of the reduction of \cite{Cohen-AddadSL22} from the multilayered PCPs of \cite{DGKR05, Khot02} to $3$-Hypergraph Vertex Coverage.

\textbf{Nearly Linear Time PTAS.} An important feature of $\ell_2^2$ Min-Sum Clustering is that we can use assignments of clusters to their mean to obtain the cost of the points in the cluster, an idea previously used in \cite{Indyk99,Matousek00,Schulman00,VegaKKR03}. We show how to reduce the number of candidate means to a constant (depending only on $k$ and $\varepsilon$.
The idea here is to use $D^2$ sampling methods akin to $k$-means++ \cite{ArthurV07}. Unfortunately, by itself, it is not sufficient as there may exist clusters that have significant min-sum clustering cost, but are not detectable by $D^2$ sampling. To this end, we augment $D^2$ sampling via a careful pruning strategy that removes high costing points, increasing the relative cost of clusters of high density. 
Thereafter, we show that given sufficiently many samples, we can find a small set of suitable candidate means that are induced by a nearly optimal clustering.

What remains to be shown is how to find an assignment of points to these centers with similar cost. For this, we could use a flow-based approach, but this results in a $n^3$ running time. Instead, we employ a discretization and bucketing strategy that allows us to sparsify the point set while preserving the min-sum clustering cost, akin to coresets.

\paragraph{Learning-augmented algorithm.}
Our starting point for our learning-augmented algorithm for min-sum $k$-clustering is the learning-augmented algorithms for $k$-means clustering by \cite{ErgunFSWZ22,NguyenCN23}. 
The algorithms note that the $k$-means clustering objective can be decomposed across the points that are given each label $i\in[k]$. 
Thus we consider the subset $P_i$ of points of the input dataset $X$ that are given label $i$ by the oracle. 
Since $k$-means clustering objective can be further decomposed along the $d$ dimensions, then the algorithms consider $P_i$ along each dimension. 

The cluster $P_i$ can have an $\alpha$ fraction of incorrect points. 
The main observation is that there can be two cases. 
Either $P_i$ includes a number of ''bad'' points that are far from the true mean and thus easy to identify, or $P_i$ includes a number of ``bad'' points that are difficult to identify but also are close to the true mean and thus do not largely affect the overall $k$-means clustering cost. 
Thus the algorithm simply needs to prune away the points that are far away, which can be achieved by selecting the interval of $(1-\O{\alpha})$ points that has the best clustering cost. 
It is then shown that the resulting centers provide a good approximate solution to the $k$-means clustering cost. 

Unfortunately, we cannot immediately utilize the previous approach because min-sum $k$-clustering is a density-based clustering rather than a centroid-based clustering. 
However, it is known~\cite{InabaKI94} that we can rewrite
\[\sum_{i\in[k]}\sum_{p,q\in C_i}\|p-q\|_2^2=\sum_{i\in[k]}|C_i|\cdot\sum_{p\in C_i}\|p-c_i\|_2^2,\]
where $c_i$ is the centroid of the points in the cluster $C_i$ in an approximately optimal clustering $\mathcal{C}=\{C_1,\ldots,C_k\}$. 
We can use the learning-augmented $k$-means clustering algorithm to identify good proxies for each centroid $c_i$. 
Moreover, by our assumptions on the precision and recall of each cluster, we have that $|P_i|$ is a good estimate of $|C_i|$. 
Therefore, we have a good approximation of the cost of the optimal min-sum $k$-clustering; it remains to identify the actual clusters. 

In standard centroid-based clustering, each point is assigned to its closest center. 
However, this is not true for min-sum $k$-clustering. 
Thus, we seek alternative approaches to identifying a set of approximately $|P_i|$ to each centroid returned by the learning-augmented $k$-means algorithm. 
To that end, we define a constrained min-cost flow problem as follows. 
We create a source node $s$ and a sink node $t$, requiring $n=|X|$ flow from $s$ to $t$. 
We then create a directed edge from $s$ to each node $u_x$ representing a separate $x\in X$ with capacity $1$ and cost $0$. 
These two gadgets ensure that a unit of flow must be pushed across each node representing a point in the input dataset. 

We also create a directed edge to $t$ from each node $v_i$ representing a separate $c_i$ with capacity $\frac{1}{1-\alpha}\cdot|P_i|$ and cost $0$. 
For each $x\in X$, $i\in[k]$, create a directed edge from $u_x$ to $v_i$ with capacity $1$ and cost $\frac{1}{1-\alpha}\cdot|P_i|\cdot\|x-c_i\|_2^2$. 
These two gadgets ensure that when a flow is pushed across some node to the corresponding node representing a center, then the cost of the flow is almost precisely the cost of assigning a point to the corresponding center toward the min-sum $k$-clustering objective. 
Finally, we require that at least $(1-\alpha)\cdot|P_i|$ flow goes through node $v_i$ correpsonding to center $c_i$. 
This ensures that the correct number of points is assigned to each center consistent with the precision and recall assumptions. 

We note that the constrained min-cost flow problem can be written as a linear program. 
Therefore to identify the overall clusters, we run any standard polynomial-time algorithm for solving linear programs~\cite{Karmarkar84,Vaidya89a,Vaidya90,LeeS15,LeeSZ19,CohenLS21,JiangSWZ21}. 
It then follows by that well-known integrality theorems for min-cost flow, the resulting solution is integral and thus provides a valid clustering with approximately optimal $\ell_2^2$ min-sum $k$-clustering objective. 

\subsection{Related Works}
\seclab{sec:related}
The min-sum $k$-clustering problem was first introduced for general graphs by \cite{SahniG76}. 
The problem is complement of the max $k$-cut problem, in which the goal is to partition the vertices of an input graph into $k$ subsets as to maximize the number or weight of the edges crossing any pair of subsets, c.f.,~\cite{PapadimitriouY91}.  
\cite{Guttmann-BeckH98a} showed that the $\ell_2$ min-sum $k$-clustering problem is also closely related to the balanced $k$-median problem, in which the goal is to identify $k$ centers $c_1,\ldots,c_k$ and partition the input dataset $X$ into clusters $C_1,\ldots,C_k$ to minimize $\sum_{i=1}^k|C_i|\sum_{x\in X}\|x-c_i\|_2$.
In particular, \cite{Guttmann-BeckH98a} showed that an $\alpha$-approximation to balanced $k$-median yields a $2\alpha$-approximation to min-sum $k$-clustering. 
\cite{Guttmann-BeckH98a} then showed that balanced $k$-median can be solved in time $n^{\O{k}}$ by guessing the cluster centers and sizes, and then subsequently determining the assignment between the input points and the centers, which also results in a $2$-approximation for min-sum $k$-clustering in $n^{\O{k}}$ time. 
For the structurally different $\ell_2$ min-sum $k$-clustering problem, \cite{BehsazFSS19} achieved a polynomial-time algorithm that achieves the best known approximation of $\O{\log n}$, by considering the embedding of metric spaces into hierarchically separated trees using dynamic programming. 
However, these techniques do not immediately translate into a good approximation for $\ell_2^2$ min-sum $k$-clustering. 
Even more recently, \cite{naderi2024approximation} provided a QPTAS in metrics induced by graphs of bounded treewidth, and graphs of bounded doubling dimension.

For the prize-collecting version of $\ell_2$ min-sum $k$-clustering, \cite{HassinO10} gave a $2$-approximation algorithm in the metric setting that uses polynomial time for fixed constant $k$. 
In a separate line of work, \cite{BalcanB09,BalcanBG09} address conditions under which the clustering would be stable. 
Namely for the metric case and small $k$, they compute a clustering that is to the optimal $\ell_2$ min-sum $k$-clustering in the sense that most of the labels are correct, though the objective value may not be close to the optimal value. 

On the lower bound side,  \cite{Guttmann-BeckH98a} showed that the general min-sum $k$-clustering problem is NP-hard, while \cite{AloiseDHP09} showed that even the $\ell_2^2$ min-sum $k$-clustering problem is NP-hard even when $k=2$. 
\cite{KannKLP97} first showed that it is NP-hard to approximate non-metric min-sum $k$-clustering within a multiplicative $\O{n^{2-\eps}}$-factor for any $\eps>0$ and $k>3$. 
Recently, \cite{Cohen-AddadSL21} showed that for metric min-sum $k$-clustering, it is NP-hard to approximate within a multiplicative $1.415$-factor. 
However, prior to this work, no such hardness-of-approximation was known for the $\ell_2^2$ min-sum $k$-clustering problem. 

A popular way of obtaining polynomial time approximation schemes are coresets, which are succinct summaries of a data set with respect to a given clustering objective.
For $\ell_2^2$ min-sum clustering, the most closesly related construction is the classic $k$-means problem, as well as variants such as non-uniform $k$-clustering.
Following a long line of work~\cite{BecchettiBC0S19,FeldmanSS20,HuangV20,Cohen-AddadWZ23,WoodruffZZ23}, a $k$-means coreset in Euclidean space of size $\tilde{O}(\frac{k}{\eps^2}\cdot\min\left(\frac{1}{\eps^2},\sqrt{k}\right))$ is known to exist~\cite{Cohen-AddadSS21,Cohen-AddadLSSS22,BansalCPSS24}, which was surprisingly shown to be optimal~\cite{HuangLW24,ZhuTHH24}.

For non-uniform clustering, centers are associated with weights and the clustering cost is $\sum_{i=1}^k \sum_{p\in C_i} w_{c_i}\|p-c_i\|^2$, where $c_i$ is the center associated with the cluster $C_i$ and $w_{c_i}$ denotes its weight. Min-sum clustering is a related problem where the weight is not arbitrary, but chosen to be equal to $|C_i|$. Unfortunately, the only known coreset constructions for the weighted $k$-means problem \cite{FeldmanS12} only apply to the line metric and even in this case have size at least $(\log n)^k$.
Nevertheless, coreset based approaches have been successfully used to obtain fast algorithms with additive errors in general metric spaces, see \cite{CzumajS07}. It is unclear if these ideas can improve algorithms for $\ell_2^2$ min-sum clustering, even when using additive errors.

\subsection{Preliminaries}
We use the notation $[n]$ to denote the set $\{1,2,\ldots,n\}$ for an integer $n>0$. 
For a set $X$, we use the notation $X=A\dot\cup B$ to denote that $A$ and $B$ partition $X$, i.e., $A\cup B=X$ and $A\cap B=\emptyset$. 
For a matrix $A\in\mathbb{R}^{n\times d}$, we define its Frobenius norm as
\[\|A\|_F:=\sqrt{\sum_{i=1}^n\sum_{j=1}^d A_{i,j}^2}.\]
We use $\poly(n)$ to denote a fixed polynomial in $n$ whose degree can be determined by setting appropriate constants in the algorithms or proofs. 
We use $\polylog(n)$ to denote $\poly(\log n)$. 
For a function $f(\cdot,\ldots,\cdot)$, we use the notation $\tO{f}$ to denote $f\cdot\polylog(f)$. 

\paragraph{$k$-means clustering.}
In the Euclidean $k$-means clustering problem, the input is a dataset $X\subset\mathbb{R}^d$ and the goal is to partition $X$ into clusters $C_1,\ldots,C_k$ by assigning a centroid $c_i$ to each cluster $C_i$ as to minimize the objective
\[\min_{c_1,\ldots,c_k}\sum_{x\in X}\min_{i\in[k]}\|x-c_i\|_2^2.\]

\section{Hardness of Approximation of \texorpdfstring{$\ell_2^2$}{L22} Min-Sum \texorpdfstring{$k$}{k}-Clustering}
In this section, we show the hardness of approximation of $\ell_2^2$ min-sum $k$-clustering, i.e., \thmref{thm:mainlb}. 
We first define the relevant formulations of Johnson Coverage Hypothesis in \secref{sec:jch}. Next, in \secref{sec:inapprox} we provide the main reduction from the Johnson coverage problem to the $\ell_2^2$ min-sum $k$-clustering problem. Finally, in \secref{sec:jchdb} we prove a special case of a generalization of \JCHDB which yields the unconditional NP-hardness factor claimed in \thmref{thm:mainlb}.
\subsection{Johnson Coverage Hypothesis}
\seclab{sec:jch}
In this section, we recall the \jc problem, followed by the \jc hypothesis \cite{Cohen-AddadSL22}.

Let $n,z,y\in\mathbb N$ such that $n\ge z > y$. Let $E\subseteq \binom{[n]}{z}$ and $S\in \binom{[n]}{y}$. We define the coverage of $S$ w.r.t.\ $E$, denoted by $\con(S,E)$ as follows:
$$
\con(S,E)=\{T\in E\mid S\subset T\}.
$$

\begin{definition}[Johnson Coverage Problem]
In the $(\alpha,z,y)$-Johnson Coverage problem with $z > y \geq 1$, we are given a  universe $U:=[n]$, a collection of subsets of $U$, denoted by $E\subseteq \binom{[n]}{z}$, and a parameter $k$ as input. We would like to distinguish between the following two cases:
\begin{itemize}
\item \textbf{\emph{Completeness}}: There exists $\C:=\{S_1,\ldots ,S_k\}\subseteq  \binom{[n]}{y}$ such that 
$$\con(\C):=\underset{i\in[k]}{\cup} \con(S_i,E)=E.$$
\item \textbf{\emph{Soundness}}: For every $\C:=\{S_1,\ldots ,S_k\}\subseteq \binom{[n]}{y}$ we have ${\left|\con(\C)\right|}\le \alpha\cdot\left|E\right|$.
\end{itemize}
\label{def:gjc}
We call $(\alpha, z, z-1)$-Johnson Coverage as $(\alpha, z)$-Johnson Coverage. 
\end{definition}

Notice that $(\alpha,2)$-Johnson Coverage Problem is simply the well-studied vertex coverage problem (with gap $\alpha$). Also, notice that if instead of picking the collection $\C$ from $\binom{[n]}{y}$, we replace it with picking the collection $\C$ from $\binom{[n]}{1}$ with a similar notion of coverage, then we simply obtain the  Hypergraph Vertex Coverage problem (which is equivalent to the \mkc problem for unbounded $z$). In \figref{fig:jch} we provide a few examples of instances of the Johnson coverage problem.

\begin{figure*}[!htb]
\centering
\begin{subfigure}[b]{0.3\textwidth}
\centering
\begin{tikzpicture}[scale=0.5]
\filldraw (-1,0) circle (0.1);
\filldraw (-1.7,1) circle (0.1);
\filldraw (-1.7,2) circle (0.1);
\filldraw (-1,3) circle (0.1);

\filldraw (1,0) circle (0.1);
\filldraw (1.7,1) circle (0.1);
\filldraw (1.7,2) circle (0.1);
\filldraw (1,3) circle (0.1);

\draw (-1,0) -- (0,0.8);
\draw (-1.7,1) -- (0,0.8);
\draw (1,0) -- (0,0.8);
\draw (1.7,1) -- (0,0.8);

\draw (-1.7,2) -- (0,2.2);
\draw (-1,3) -- (0,2.2);
\draw (1.7,2) -- (0,2.2);
\draw (1,3) -- (0,2.2);

\draw (0,0.8) -- (0,2.2);

\filldraw[red] (0,0.8) circle (0.2);
\filldraw[red] (0,2.2) circle (0.2);
\end{tikzpicture}
\caption{}
\figlab{fig:jch:one:two}
\end{subfigure}
\begin{subfigure}[b]{0.3\textwidth}
\centering
\begin{tikzpicture}[scale=0.5]
\draw [rounded corners,fill=gray!90,opacity=0.5] (-1,3.4)--(-2.2,1.7)--(0.5,2.1)--cycle;
\draw [rounded corners,fill=gray!90,opacity=0.5] (1,3.4)--(2.2,1.7)--(-0.5,1.9)--cycle;
\draw [rounded corners,fill=gray!90,opacity=0.5] (-1,-0.6)--(-2.3,1.4)--(0.6,1)--cycle;
\draw [rounded corners,fill=gray!90,opacity=0.5] (1,-0.6)--(2.3,1.4)--(-0.5,1)--cycle;

\filldraw (-1,0) circle (0.1);
\filldraw (-1.7,1) circle (0.1);
\filldraw (-1.7,2) circle (0.1);
\filldraw (-1,3) circle (0.1);

\filldraw (0,0.8) circle (0.1);
\filldraw (0,2.2) circle (0.1);

\filldraw (1,0) circle (0.1);
\filldraw (1.7,1) circle (0.1);
\filldraw (1.7,2) circle (0.1);
\filldraw (1,3) circle (0.1);

\filldraw[red] (0,0.8) circle (0.1);
\filldraw[red] (0,2.2) circle (0.1);
\end{tikzpicture}
\caption{}
\figlab{fig:jch:one:three}
\end{subfigure}
\begin{subfigure}[b]{0.3\textwidth}
\centering
\begin{tikzpicture}[scale=0.5]
\draw [rounded corners,fill=gray!90,opacity=0.5] (-1,3.4)--(-2.2,1.7)--(0.5,2.1)--cycle;
\draw [rounded corners,fill=gray!90,opacity=0.5] (1,3.4)--(2.2,1.7)--(-0.5,1.9)--cycle;
\draw [rounded corners,fill=gray!90,opacity=0.5] (-1,-0.6)--(-2.3,1.4)--(0.6,1)--cycle;
\draw [rounded corners,fill=gray!90,opacity=0.5] (1,-0.6)--(2.3,1.4)--(-0.5,1)--cycle;

\filldraw (-1,0) circle (0.1);
\filldraw (-1.7,1) circle (0.1);
\filldraw (-1.7,2) circle (0.1);
\filldraw (-1,3) circle (0.1);

\filldraw (0,0.8) circle (0.1);
\filldraw (0,2.2) circle (0.1);

\filldraw (1,0) circle (0.1);
\filldraw (1.7,1) circle (0.1);
\filldraw (1.7,2) circle (0.1);
\filldraw (1,3) circle (0.1);

\draw (-1,0) -- (0,0.8);
\draw (-1.7,1) -- (0,0.8);
\draw (1,0) -- (0,0.8);
\draw (1.7,1) -- (0,0.8);

\draw (-1.7,2) -- (0,2.2);
\draw (-1,3) -- (0,2.2);
\draw (1.7,2) -- (0,2.2);
\draw (1,3) -- (0,2.2);

\draw (0,0.8) -- (0,2.2);
\end{tikzpicture}
\caption{}
\figlab{fig:jch:two:three}
\end{subfigure}
\caption{Examples of input instances of the Johnson Coverage Hypothesis for $k=2$. 
\figref{fig:jch:one:two} shows an example of a completeness instance of $\left(0.7,2,1\right)$, since all subsets of size $2$, i.e., all edges, can be covered by $k=2$ choices of subset of size $1$, i.e., two vertices.
\figref{fig:jch:one:three} shows an example of a completeness instance of $\left(0.7,3,1\right)$, since all subsets of size $3$ can be covered by $k=2$ vertices. 
\figref{fig:jch:two:three} shows an example of a soundness instance of $\left(0.7,3,2\right)$, since at most $2\le 0.7\cdot 4$ subsets of size $3$ can be covered by any choice of $k=2$ edges.  
}
\figlab{fig:jch}
\end{figure*}
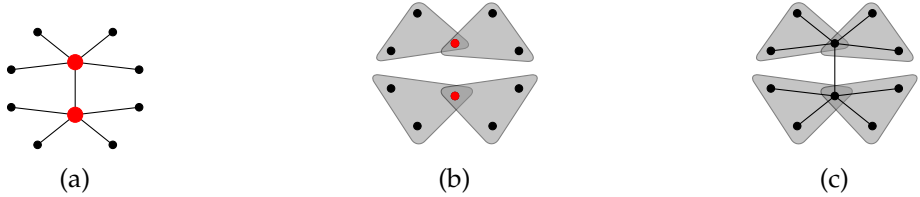

We now state the following hypothesis. 

\begin{hypothesis}[Johnson Coverage Hypothesis (\JCH)~\cite{Cohen-AddadSL22}]\hyplab{hyp:jch}
For every constant $\varepsilon>0$, there exists a constant $z:=z(\varepsilon)\in\mathbb{N}$ such that deciding the $\left(1-\frac{1}{e}+\varepsilon,z\right)$-Johnson Coverage Problem is \NP-Hard. 
\end{hypothesis}

Note that since Vertex Coverage problem is a special case of the \jc problem, we have that the \NP-Hardness of $(\alpha,z)$-Johnson Coverage problem is already known for $\alpha=0.944$ \cite{AKS11} (under unique games conjecture). 

On the other hand,  if we replace picking the collection $\C$ from $\binom{[n]}{z-1}$ by picking from $\binom{[n]}{1}$, then for the Hypergraph Vertex Coverage problem, we do know that for every $\eps>0$ there is some constant $z$ such that the Hypergraph Vertex Coverage problem is \NP-Hard to decide for a factor of $\left(1-\frac{1}{e}+\varepsilon\right)$ \cite{F98}. 

For continuous clustering objectives, a dense version of \JCH is sometimes needed to prove inapproximability results (see \cite{Cohen-AddadSL22} for a discussion on this). Thus, we state: 

\begin{hypothesis}[Dense Johnson Coverage Hypothesis (\JCHD)~\cite{Cohen-AddadSL22}]
\JCH holds for instances $(U,E,k)$ of \jc problem where $|E|=\omega(k)$.\hyplab{hyp:jchd}
\end{hypothesis}

More generally, let $(\alpha, z, y)$-\jcd problem be the special case of the 
$(\alpha, z, y)$-\jc problem where the instances satisfy $|E| = \omega(k \cdot |U|^{z - y - 1})$. 
Then \JCHD states that for any $\eps > 0$, there exists $z = z(\eps)$ such that 
$(1-1/e + \eps, z, z-1)$-\jcd is \NP-Hard. 
This additional property has always been obtained in literature by looking at the hard instances that were constructed. 
In \cite{Cohen-AddadS19}, where the authors proved the previous best inapproximability results for continuous case $k$-means and $k$-median, 
it was observed that hard instances of $(0.94,2,1)$-\jc constructed in~\cite{AKS11} can be made to satisfy the above property.

Now we are ready to define the variant of \JCH needed for proving inapproximability of $\ell_2^2$ min-sum $k$-clustering.
For any two non-empty finite sets $A,B$, and a constant $\delta\in [0,1]$, we say a function $f:A\to B$ is $\delta$-balanced if for all $b\in B$ we have: 
$$  \left|\{a\in A: f(a)=b\}\right| \le (1+\delta)\cdot \frac{|A|}{|B|}.
$$

We then put forth the following hypothesis.

\begin{hypothesis}[Dense and Balanced Johnson Coverage Hypothesis (\JCHDB)]\hyplab{hyp:jchdb}
\JCH holds for instances $(U,E,k)$ of \jc problem where $|E|=\omega(k)$ and in the completeness case there   exists $\C:=\{S_1,\ldots ,S_k\}\subseteq  \binom{[n]}{z-1}$  and a {\bf 0-balanced} function $\psi:E\to [k]$ such that for all $T\in E$ we have $S_{\psi(T)}\subset T$.
\end{hypothesis}

More generally, let $(\alpha, z, y,\delta)$-\jcdb problem be the special case of the 
$(\alpha, z, y)$-\jcd problem where the instances admit a $\delta$-balanced function $\psi:E\to [k]$ in the completeness case which  partitions $E$ to $k$ parts, say $E_1\dot\cup \cdots \dot\cup E_k$ such that for all $i\in[k]$ we have $\con(S_i,E_i)=E_i$ and $|E_i|\le \frac{|E|}{k}\cdot (1+\delta)$. 
Then \JCHDB states that for any $\eps > 0$, there exists $z = z(\eps)$ such that 
$(1-1/e + \eps, z, z-1,0)$-\jcdb is \NP-Hard.

As with the case of \JCHD, the balanced addition to \JCHD is also quite natural and candidate constructions typically give this property for free. To support this point, we will prove some special case of this.

In \cite{Cohen-AddadSL22} the authors had proved the following special case of \JCHD.

\begin{theorem}[\cite{Cohen-AddadSL22}]
\thmlab{thm:CKLNP}
For any $\eps > 0$, given a simple 3-hypergraph $\mathcal{H} = (V, H)$ with $n = |V|$, it is NP-hard to distinguish between the following two cases:
\label{thm:np-hard}
\begin{itemize}
\item {\bf Completeness:} There exists $S \subseteq V$ with $|S| = n/2$ that intersects every hyperedge. 
\item {\bf Soundness:} Any subset $S \subseteq V$ with $|S| \leq n/2$ intersects at most a $(7/8 + \eps)$ fraction of hyperedges. 
\end{itemize}
Furthermore, under randomized reductions, the above hardness holds when $|H| = \omega(n^2)$. 
\label{thm:hvc}
\end{theorem}

In \secref{sec:jchdb}, we further analyze the proof of the above theorem and prove the following:

\begin{theorem}
\thmlab{thm:minsumNP}
\thmref{thm:CKLNP} holds even with the following additional completeness guarantee for all $\delta>0$: there   exists $S:=\{v_1,\ldots ,v_k\}\subseteq  V$  and a $\delta$-balanced function $\psi:H\to [k]$ such that for all $e\in H$ we have $v_{\psi(e)}\in e$.
\end{theorem}

This result will be used to prove the unconditional NP-hardness of approximating $\ell_2^2$ min-sum $k$-clustering problem.

\subsection{Inapproximability of \texorpdfstring{$\ell_2^2$}{L22} min-sum \texorpdfstring{$k$}{k}-clustering}
\seclab{sec:inapprox}
 
In the following subsection, we will prove the below theorem. 

\begin{theorem}
\thmlab{thm:mainreduction}
Assume $(\alpha, z, y,\delta)$-\jcdb is \NP-Hard. 
For every constant $\varepsilon>0$,  given a point-set $P\subset \mathbb{R}^{d}$ of size $n$ (and $d=\O{\log n}$) and a parameter $k$ as input, it is \NP-Hard to distinguish between the following two cases:
\begin{itemize}
\item \textbf{\emph{Completeness}}:  
There exists partition $P_1^*\dot\cup\cdots \dot\cup P_k^*:=P$  such that $$\sum_{i\in [k]}\sum_{p,q\in P_i^*}\|p-q\|_2^2\le (1+3\delta)\cdot (z-y)\cdot \rho n^2/k,$$
\item \textbf{\emph{Soundness}}: For every partition $P_1\dot\cup\cdots \dot\cup P_k:=P$  we have  $$\sum_{i\in [k]}\sum_{p,q\in P_i}\|p-q\|_2^2\ge (1-o(1))\cdot \left(\alpha\cdot \sqrt{z-y}+(1-\alpha)\cdot \sqrt{z-y+1}\right)^{2} \cdot \rho n^2/k,$$
\end{itemize}
for some constant $\rho>0$.
\end{theorem}

Putting together the above theorem with \thmref{thm:minsumNP} (i.e., NP-hardness of $(7/8+\varepsilon,3,1,\delta)$-\jcdb problem for all $\varepsilon,\delta>0$), we obtain the NP-hardness of approximating $\ell_2^2$ min-sum $k$-clustering. The above theorem also immediately yields the hardness of approximating $\ell_2^2$ min-sum $k$-clustering under \JCHDB (i.e., conditional NP-hardness of $(1-1/e+\varepsilon,z,z-1,0)$-\jcdb problem for all $\varepsilon>0$ and some $z=z(\varepsilon)\in\mathbb{N}$). This completes the proof of~\thmref{thm:mainlb}.

\subsubsection{Proof of \texorpdfstring{\thmref{thm:mainreduction}}{Theorem (Main Reduction)}}
Fix $\varepsilon>0$ as in the theorem statement.  Let $\eps':=\eps/11$. Starting from a hard instance of $\left( \alpha, z, y, \delta \right)$-\jcdb problem $(U,E,k)$ with $|U| = n$ and $|E| = \omega(n^{z-y})$, 

\paragraph{Construction.} 
The $\ell_2^2$ min-sum $k$-clustering instance consists of the set of points to be clustered $P\subseteq \{0,1\}^{n}$ where for every $T\in E$ we have the point $p_t\in P$ defined as follows:
\[p_T:=\sum_{i\in T}\vec{e}_i\].
From the construction, it follows that for every distinct $T,T'\in E$, we have:
\begin{align}\label{eqsub}
 \|p_T-p_{T'}\|_2^2= 2z -2\cdot |T\cap T'|. 
\end{align}

\paragraph{Completeness.}
Suppose there exist $S_1,\ldots ,S_k\in \binom{[n]}{y}$   and a $\delta$-balanced function $\psi:E\to [k]$ such that for all $T\in E$ we have $S_{\psi(T)}\subset T$. 
Then, we define a clustering $C_1\dot\cup\cdots\dot\cup C_k=P$  as follows: for every $p_T\in P$, we  include it in cluster $C_{\psi(T)}$.
We now provide an upper bound on the $\ell_2^2$ min-sum cost of clustering $\mathcal{C}:=\{C_1,\ldots ,C_k\}$.
~\eqref{eqsub} implies that for each $C_i$, for any pair $T, T'$ such that
$p_T,p_{T'} \in C_i$, we have that $S_{\psi(T)}\subseteq T\cap T'$ and thus $\|p_T - p_{T'} \|^2_2 \le 2z-2y$.
Thus, the cost of clustering $\C$ is bounded as follows:
$$\sum_{i\in[k]}\sum_{p,q\in C_i}\|p-q\|_2^2\le \sum_{i\in[k]}\left(\left(|C_i|^2-|C_i|\right)\cdot 2\cdot (z-y)\right)\le 2|P|\cdot \left(\frac{|P|}{k}-1\right)\cdot (1+\delta)^2\cdot (z-y).$$
The completeness analysis is now completed by noting that $(1+\delta)^2\le 1+3\delta$. Thus we turn to the soundness analysis.

\paragraph{Soundness.}
Consider the optimal $\ell_2^2$ min-sum $k$-clustering $\mathcal{C}:=\{C_1,\ldots ,C_k\}$ of the instance (i.e., $C_1\dot\cup\cdots\dot\cup C_k=P$). 
We aim at showing that the $\ell_2^2$ min-sum $k$-clustering cost of $\C$ is at
least $((z-y) + 2(1 - \alpha) - o(1))\ell |P|$. Given a cluster $C_i$,
let $E_i := \{ T \in E : p_T\in C_i \}$ be the collection of $z$-sets of $E$ corresponding to $C_i$. 
For each
$S \in \binom{[n]}{y}$, we define the \emph{degree of $S$ in $C_i$} to be 
\[d_{i, S} := \left|\{T \mid S \subset T \text{ and } p_T \in C_i\}\right|.\]
Let $t_1 = 2z - 2y$ and $t_2 = 2z - 2y + 2$. 
For each cluster $C_i$, let
\begin{align*}
F_i &= \bigg| \{ (p, q) \in C_i^2 : \| p - q \|_2^2 \geq t_2 \} \bigg| \\
M_i &= \bigg| \{ (p, q) \in C_i^2 : \| p - q \|_2^2=t_1  \} \bigg| \\
N_i &= \bigg| \{ (p, q) \in C_i^2 : \| p - q \|_2^2 < t_1  \} \bigg|.
\end{align*}
By~\eqref{eqsub}, $F_i$, $M_i$, and $N_i$ are the number of (ordered) pairs within $C_i$ whose corresponding $z$-sets in the \jcdb instance intersect in $<y$, $=y$, and $>y$ elements respectively. 
Let $\Delta_i=\underset{S\in\binom{[n]}{y} }{\max}\ d_{i,S}$ and observe that $\Delta_i\le|C_i|$. 
We write the total cost of the clustering as follows.
\begin{align}
& \sum_{i\in[k]}\sum_{p,q\in C_i}\|p-q\|_2^2 
\geq  \sum_{i\in[k]}  \bigg( F_i t_2 + M_i t_1 \bigg)   
= \sum_{i\in[k]}  \bigg( (|C_i|^2 - M_i) t_2 + M_i t_1 - N_i t_2 \bigg)
\label{eq:cost}
\end{align}

We first upper bound $\sum_{i\in[k]} (N_i t_2) $. 
For each $z$-set $T$, there are at most $\left(\sum_{\ell=y+1}^{z} \binom{z}{\ell}\binom{n-z}{z-\ell}\right)$ many sets in $\binom{[n]}{z}$ that intersect with $T$ in at least $y + 1$ elements. 
Therefore, we have:
\begin{align*}
    \sum_{i\in[k]}N_i&\le \left(\sum_{i\in[k]} |C_i|\cdot \left(\sum_{\ell=y+1}^{z} \binom{z}{\ell}\binom{n-z}{z-\ell}\right)\right)\\
    &\le \sum_{i\in[k]}|C_i|\cdot 2^z\cdot (z-y)\cdot n^{z-y-1}\\
    &=\O{|P|\cdot n^{z-y-1}}.
\end{align*}

By the definition of \jcdb, $|P|=|E| = \omega(k \cdot n^{z-y-1})$, 
so $\sum_{i\in[k]} N_i t_2 $ is at most $o( |P|^2/k)$.

Next, we invoke a technical claim in \cite{Cohen-AddadSL22} which bounds $M_i / |C_i|$ in terms of $\Delta_i$ and $|C_i|$. 
\begin{claim}[Claim 3.18 in \cite{Cohen-AddadSL22}]
For every $i \in [k]$, either $|C_i| = o(|P| / k)$ or $M_i / |C_i| \leq (1 + o(1)) \Delta_i + o(|C_i|)$. 
\label{claim:sumdegsquared}
\end{claim}

We can thus lower bound the cost of the clustering in \eqref{eq:cost} as follows:
\begin{align}
 \sum_{i\in[k]}\sum_{p,q\in C_i}\|p-q\|_2^2 
&\ge \left(\sum_{i\in[k]}   |C_i|^2 (2z-2y+2) \right) -  \left(\sum_{i\in[k]} 2\Delta_i|C_i|    \right) - o\left(|P^2|/k+\sum_{i\in[k]}   |C_i|^2\right)\label{eq:cost1}
\end{align}

Thus, we now look at upper bounding $\sum_{i\in[k]} 2\Delta_i|C_i|$. From the soundness case assumption, we have that $s:=\sum_{i\in [k]}\Delta_i \le \alpha\cdot |E|$. Without loss of generality, we may assume that $|C_1|\ge |C_2|\ge \cdots \ge |C_k|$. Let $t\in[k]$ be smallest integer such that  $\sum_{i\in[t]}|C_i|>s$. 
Since $\Delta_i\le|C_i|$, then $\sum_{i\in[k]} 2\Delta_i|C_i|$ is maximized when $\Delta_i=|C_i|$ for all $i\in [t]$. Thus, $\sum_{i\in[k]} 2\Delta_i|C_i|\le \sum_{i\in[t]} 2|C_i|^2$, and we can rewrite \eqref{eq:cost1} as follows:
\begin{align}
 \sum_{i\in[k]}\sum_{p,q\in C_i}\|p-q\|_2^2 
&\geq  \left(\sum_{i\in[t]}   |C_i|^2 (2z-2y) \right) +\left(\sum_{i=t+1}^k   |C_i|^2 (2z-2y+2) \right)  - o\left(|P^2|/k+\sum_{i\in[k]}   |C_i|^2\right) \label{eq:cost2}
\end{align} 
Then the quantity $\left(\sum_{i\in[t]}   |C_i|^2 (2z-2y) \right) +\left(\sum_{i=t+1}^k   |C_i|^2 (2z-2y+2) \right)$ is minimized when for all $i\in[t]$, we have all $|C_i|$'s to be equal and for all $i\in\{t+1,\ldots ,k\}$, we have all $|C_i|$'s to be equal (by convexity). Thus, $$\left(\sum_{i\in[t]}   |C_i|^2 (z-y) \right) +\left(\sum_{i=t+1}^k   |C_i|^2 (z-y+1) \right)\ge \left(   \frac{\alpha^2 |P|^2}{t} (z-y) \right) +\left(   \frac{(1-\alpha)^2|P|^2}{(k-t)} (z-y+1) \right).$$

We may rewrite the left side as follows:
$$ \frac{|P|^2}{k}\left(   \frac{\alpha^2\cdot (z-y) }{t/k}   +    \frac{(1-\alpha)^2\cdot (z-y+1)}{1-(t/k)}  \right)$$

If we look at the first derivative of the above expression w.r.t. $t/k$, then we have that the minima of the above expression is attained when:
$$\frac{(1-\alpha)^2\cdot (z-y+1)}{(1-(t/k))^2}=\frac{\alpha^2\cdot (z-y)}{(t/k)^2}$$
Simplifying, we obtain:
$$t=k\cdot \left(\frac{\alpha\cdot \sqrt{z-y}}{\alpha\cdot \sqrt{z-y}+(1-\alpha)\cdot \sqrt{z-y+1}}\right)$$

Returning to the cost of clustering, we have from \eqref{eq:cost2}:
\begin{align*}
 \sum_{i\in[k]}\sum_{p,q\in C_i}\|p-q\|_2^2 
&\geq   (1-o(1))\cdot  \frac{2|P^2|}{k} \cdot \left(   \frac{\alpha^2\cdot (z-y) }{t/k}   +    \frac{(1-\alpha)^2\cdot (z-y+1)}{1-(t/k)}  \right)\\
&\geq (1-o(1))\cdot  \frac{2|P^2|}{k} \cdot \left(   \frac{\alpha^2\cdot (z-y) }{t/k}   +  \left(\frac{\alpha^2\cdot (z-y) }{t/k} \cdot\frac{1-(t/k)}{t/k}\right) \right)\\
&=  (1-o(1))\cdot  \frac{2|P^2|}{k} \cdot \left(  \frac{\alpha^2\cdot (z-y) }{  (t/k)^2}  \right)\\
&\ge (1-o(1))\cdot  \frac{2|P^2|}{k} \cdot    \left(\alpha\cdot \sqrt{z-y}+(1-\alpha)\cdot \sqrt{z-y+1}\right)^{2} 
\end{align*} 

\paragraph{Dimensionality reduction.}
The proof of the theorem with the reduced dimension (i.e., $d=O(\log n)$) of the hard instances follows from the Johnson-Lindenstrauss lemma. 
Elaborating, given a set of $n$ points in $\mathbb{R}^d$, we have that the $\ell_2^2$ min-sum $k$-clustering cost of a given partition $\{C_1,\ldots,C_k\}$   expressed as $\sum_{i=1}^k   \sum_{p,q \in C_i} \|p-q\|^2_2$. 
Thus, applying the Johnson-Lindenstrauss lemma with target dimension $\O{\log n/\eps^2}$ for small enough $\eps$, yields an instance where the $\ell_2^2$ min-sum $k$-clustering cost of any clustering $C$ is within a factor $(1+\eps)$ of the $\ell_2^2$ min-sum $k$-clustering cost of $C$ in the original $d$-dimensional instance. 
It follows that the gap is preserved up to a $(1+\eps)$ factor and the theorem follows. 
Note that this can be made deterministic (for example, see the result of Engebretsen et al.~\cite{EIO02}). 

\subsection{Proof of \texorpdfstring{\thmref{thm:minsumNP}}{Theorem (Min Sum NP)}}\seclab{sec:jchdb}


\thmref{thm:minsumNP} follows from observing additional properties of the reduction of \cite{Cohen-AddadSL22} from the multilayered PCPs of \cite{DGKR05, Khot02} to $3$-Hypergraph Vertex Coverage. The description of the reduction is taken verbatim from \cite{Cohen-AddadSL22}. We first describe the multilayered PCPs that we use. 

\begin{definition}
An $\ell$-layered PCP $\calm$ consists of
\begin{itemize}
\item An $\ell$-partite graph $G = (V, E)$ where $V = \cup_{i=1}^{\ell} V_i$. Let $E_{i,j} = E \cap (V_i \times V_j)$. 
\item Sets of alphabets $\Sigma_1, \dots, \Sigma_{\ell}$. 
\item For each edge $e = (v_i, v_j) \in E_{i,j}$, a surjective projection $\pi_e : \Sigma_j \to \Sigma_i$. 
\end{itemize}
Given an assignment $(\sigma_i : V_i \to \Sigma_i)_{i \in [\ell]}$, an edge $e = (v_i, v_j) \in E_{i,j}$ is {\em satisfied} if $\pi_e(\sigma_j(v_j)) = \sigma_i(v_i)$. 
There are additional properties that $\calm$ can satisfy.
\begin{itemize}
\item $\eta$-smoothness: For any $i < j$, $v_j \in V$, and $x, y \in \Sigma_j$, $\Pr_{(v_i, v_j) \in E_{i,j}} [\pi_{(v_i, v_j)}(x) = \pi_{(v_i, v_j)}(y)] \leq \eta$. 
\item Path-regularity: Call a sequence $p = (v_1, \dots, v_{\ell})$ {\em full path} if $(v_i, v_{i+1}) \in E_{i, i+1}$ for every $1 \leq i < \ell$,
and let $\calp$ be the distribution of full paths obtained by (1) sampling a random vertex $v_1 \in V_1$ and (2) for $i = 2, \dots, \ell$, sampling $v_i$ from the neighbors of $v_{i-1}$ in $E_{i-1, i}$. 
$\calm$ is called {\em path-regular} if for any $i < j$, sampling $p = (v_1, \dots, v_{\ell})$ from $\calp$ and taking $(v_i, v_j)$ is the same as sampling uniformly at random from $E_{i,j}$. 
\end{itemize}
\end{definition}

\begin{theorem}
\thmlab{thm:dgkr}
\cite{DGKR05, Khot02}
For any $\tau, \eta > 0$ and $\ell \in \NN$, given an $\ell$-layered PCP $\calm$ with $\eta$-smoothness and path-regularity, it is NP-hard to distinguish between the following cases.
\begin{itemize}
\item {\bf Completeness: } There exists an assignment that satisfies every edge $e \in E$.
\item {\bf Soundness: } For any $i < j$, no assignment can satisfy more than an $\tau$ fraction of edges in $E_{i,j}$. 
\end{itemize}
\end{theorem}

Given an $\ell$-layered PCP $\calm$ described above, in \cite{Cohen-AddadSL22} they design the reduction to the Johnson Coverage problem as follows.
First, the produced instance will be vertex-weighted and edge-weighted, so that the problem becomes ``choose a set of vertices of total weight at most $k$ to maximize the total weight of covered edges.'' We will explain how to obtain an unweighted instance at the end of this section.

\begin{itemize}
\item Let $C_i := \{ \pm 1 \}^{|\Sigma_i|}$ and $U_i := V_i \times C_i$. 
The resulting hypergraph will be denoted by $\calh = (U, H)$ where $U = \cup_{i=1}^{\ell} (V_i \times C_i)$. The weight of vertex $(v, x) \in V_i \times C_i$ is 
\[
w(v, x) :=\frac{1}{\ell} \cdot \frac{1}{|V_i|} \cdot \frac{1}{|C_i|}.
\]
Note that the sum of all vertex weights is $1$. 

\item Let $\cald_I$ be the distribution where $i \in [\ell]$ is sampled with probability\footnote{\cite{Cohen-AddadSL22} states $(\ell-i)^2 / (6\ell(\ell - 1)(2\ell - 1))$, which is a typo corrected in their analysis.} $6(\ell-i)^2 / (\ell(\ell - 1)(2\ell - 1))$, and $\cald$ be the distribution over $(i, j) \in [\ell]^2$ where $i$ is sampled from $\cald_I$ and $j$ is sampled uniformly from $\{ i+1, \dots, \ell \}$. For each $i < j$, we create a set of hyperedges $H_{i,j}$ that have one vertex in $U_i$ and two vertices in $U_j$. 
Fix each $e = (v_i, v_j) \in E_{i,j}$ and a set of three vertices $t \subseteq (\{ v_i \} \times C_i) \cup (\{ v_j \} \times C_j)$. The weight $w(t)$ is $($the probability that $(i, j)$ is sampled from $\cald) \cdot (1/|E_{i,j}|) \cdot ($the probability that $t$ is sampled from the following procedure$)$. The reduction is parameterized by $\delta > 0$ determined later. 
\begin{itemize}
\item For each $a \in \Sigma_i$, sample $x_a \in \{ \pm 1 \}$.
\item For each $b \in \Sigma_j$, 
\begin{itemize}
\item Sample $y_b \in \{ \pm 1 \}$. 
\item If $x_{\pi(b)} = -1$, let $z_b = y_b$ with probability $1-\delta$ and $z_b = -y_b$ otherwise.
\item If $x_{\pi(b)} = 1$, let $z_b = -y_b$. 
\end{itemize}
\item Output $\{ (v_i, x), (v_j, y), (v_j, z) \}$. 
\end{itemize}
Note that the sum of all hyperedge weights is also $1$. 
\end{itemize}

\paragraph{Soundness.} The soundness of the reduction is proved in \cite{Cohen-AddadSL22}.
\begin{lemma}[\cite{Cohen-AddadSL22}] Any subset of weight at most $1/2$ intersects hyperedges of total weight at most $7/8 + o(1)$.     
\end{lemma}

\paragraph{(Almost) regularity.}
We prove the (almost) regularity of the reduction; for every vertex, the ratio between the weight of the vertex and the total weight of the hyperedges containing it is $(3 \pm o(1))$. Note that $3$ is natural as both total vertex weights and total edge weights are normalized to $1$ and each hyperedge contains three vertices. 

Fix a vertex $(v, x)$ where $v \in V_i$ for some $i \in [\ell]$. Its vertex weight $w(v, x) = \frac{1}{\ell} \cdot \frac{1}{|V_i|} \cdot \frac{1}{|C_i|}$. 
We now consider the edge weight (described as a sampling procedure) and compute the probability that a random hyperedge contains $(v, x)$. There are two possibilities. 

\begin{itemize}
\item The hyperedge is from the $j$th layer and $i$th layer for some $j < i$. For fixed $j < i$, the probability of the pair $(j, i)$ is 
\[
\frac{6(\ell - j)^2}{\ell (\ell-1)(2\ell - 1)} \cdot \frac{1}{\ell - j}
= \frac{6(\ell - j)}{\ell (\ell-1)(2\ell - 1)}, \]
and given $(j, i)$, the probability that $v$ is contained in the sampled hyperedge is 
$\frac{2 \pm o(1)}{|V_i||C_i|}$.
(Note that the distribution of either $(v_j, y)$ or $(v_j, z)$ in the procedure is the uniform distribution on $V_i \times C_i$. The factor $2$ comes from the fact that the hyperedge samples two points from the $i$th layer; the probability that the same point is sampled twice is exponentially small and can be absorbed in the $o(1)$ term.)

\item The hyperedge is from the $i$th layer and $j$th layer for some $i < j$. For fixed $i < j$, the probability of the pair $(i, j)$ is 
\[
\frac{6(\ell - i)^2}{\ell (\ell-1)(2\ell - 1)} \cdot \frac{1}{\ell - i}
= \frac{6(\ell - i)}{\ell (\ell-1)(2\ell - 1)}, \]
\end{itemize}
Summing the above events for all $j$ values, we get
\begin{align*}
(1 \pm o(1)) &\bigg(  \big( \sum_{j=1}^{i-1} \frac{6(\ell - j)}{\ell (\ell-1)(2\ell - 1)} 
\frac{2}{|V_i||C_i|} \big) 
+
\big( \sum_{j=i+1}^{\ell} \frac{6(\ell - i)}{\ell (\ell-1)(2\ell - 1)} 
\frac{1}{|V_i||C_i|} \big) 
\bigg)  \\
=\quad& \frac{6 \pm o(1)}{\ell (\ell-1)(2\ell - 1_ |V_i||C_i|} \bigg(  \big( \sum_{j=1}^{i-1} 2(\ell - j) \big) 
+
\big( \sum_{j=i+1}^{\ell} (\ell - i)
 \big) \bigg) \\
 =\quad& \frac{6 \pm o(1)}{\ell (\ell-1)(2\ell - 1) |V_i||C_i|} \bigg(  \big( \sum_{j=1}^{i-1} 2(\ell - j) \big) 
+ \big( \ell - i \big)^2 \bigg) \\
=\quad& \frac{6 \pm o(1)}{\ell (\ell-1)(2\ell - 1) |V_i||C_i|} \bigg( 2\ell (i - 1) - i(i-1) + \ell^2 - 2 \ell i + i^2 \bigg) \\
=\quad& \frac{6 \pm o(1)}{\ell (\ell-1)(2\ell - 1) |V_i||C_i|} \bigg( \ell^2 - 2\ell + i \bigg) = \frac{3 \pm \O{1/\ell} \pm o(1)}{\ell |V_i||C_i|}. 
\end{align*}
By increasing $\ell$ to be an arbitrarily large constant, we established that the total weight of the hyperedges containing $(v, x)$ is $(3 \pm o(1))$ times its vertex weight $\frac{1}{\ell |V_i| |C_i|}$. 

\paragraph{Completeness.}
If $\calm$ admits an assignment $(\sigma_i : V_i \to \Sigma_i)_{i \in [\ell]}$ that satisfies every edge $e \in E$, 
let $S := \{ (v_i, x) : v_i \in V_i, x_{\sigma_i(v_i)} = -1 \}$. 
Fix any $e = (v_i, v_j) \in E_{i,j}$ and consider the above sampling procedure to sample $x \in \{ \pm 1 \}^{\Sigma_i}$ and $y \in \{ \pm 1 \}^{\Sigma_j}$ when $b = \sigma_j(v_j)$. 
Since $\pi_e(\sigma_j(v_j)) = \sigma_i(v_i)$, at least one of $x_{\sigma_i(v_i)}$, $y_{\sigma_j(v_j)}$, $z_{\sigma_j(v_j)}$ must be $-1$ always. 
So, $S$ intersects every hyperedge with nonzero weight.

Furthermore, an inspection of the sampling procedure reveals that for a fixed vertex $(v_i, x)$ and $j > i$, a $1/2 \pm \O{\delta}$ fraction of the hyperedges containing it has all three vertices in $S$ and a $1/2 \pm \O{\delta}$ fraction of the hyperedges containing it has only $(v_i, x)$ in $S$. 
Therefore, there must be an assignment from all the hyperedges to $S$ such that (1) a hyperedge is assigned to a vertex contained by it, and (2) every vertex is assigned a $1/2 + 1/(2\cdot 3) \pm \O{\delta} = 2/3 \pm \O{\delta}$ fraction of the hyperedges containing it (which is consistent with the fact that $S$ contains half of the vertices). 

Therefore, each vertex has almost the same ratio (up to $1 \pm o(1)$ by taking $\delta$ arbitrarily small) between its weight and the total weight of the hyperedges assigned to it. 
In order to obtain an unweighted instance, for each vertex $(v, x)$, we create a new {\em cloud} of vertices $C_{v, x}$ whose cardinality is proportional to $w(v, x)$, and replace each edge $((v_1, x_1), (v_2, x_2), (v_3, x_3))$ by all possible edges between $C_{v_1, x_1}, C_{v_2, x_2}, C_{v_3, x_3}$ (with the total weight equal to the weight of the original edge).

\section{\boldmath PTAS based on $D^2$ Sampling}

For a set $A\subset\mathbb{R}^d$, let $\mu(A) := \frac{1}{|A|}\sum_{p\in A} p$ denote its mean.
Let $\mathcal{C}=\{C_1,\ldots C_k\}$ be an optimal $k$-MinSum clustering of a point set $A$.
We use $\mu_i=\mu(C_i)$ to denote the mean of $C_i$ and we use $\Delta_i = \frac{\sum_{p\in C_i}\|p-\mu_i\|^2}{|C_i|}$ to denote the average mean squared distance of $C_i$ to $\mu_i$. We further use $C_i^\beta$ to denote the subset of $C_i$ with $\|p-\mu_i\|^2 \leq \beta\cdot \Delta_i$.
Finally, let $\OPT$ denote the cost of an optimal solution. So, $\OPT = \sum_{i=1}^k |C_i|^2\cdot\Delta_i$.

\begin{definition}
We say that $m$ is an $\varepsilon$-approximate mean of $C_i$ if $\|m-\mu_i\|^2 \leq \varepsilon\cdot \Delta_i$.
We say that a set $S\subset A$ is an $(\varepsilon,\beta)$-mean seeding set for $C_i\in \mathcal{C}$, if there exists a subset $S'\cup \{s\}\subset S$ 
with $\|s-\mu_i\|^2 \leq \beta\cdot \Delta_i$
and a weight assignment $w:S'\rightarrow \mathbb{R}_{\geq 0}$ such that 
$$
\left\|\frac{1}{\sum_{p\in S'}w(p)}\sum_{p\in S'} w(p)\cdot p - \mu_i\right\|^2 \leq \varepsilon\cdot \Delta_i.
$$
\end{definition}

We will use the following well-known identities for Euclidean means.
\begin{lemma}
\cite{InabaKI94}
\label{lem:magic}
    Let $A\subset \mathbb{R}^d$ be a set of points. Then for any $c\in \mathbb{R}^d$:
    \begin{itemize}
        \item $\sum_{p\in A} \|p-c\|^2 = \sum_{p\in A} \|p-\mu(A)\|^2 + |A|\cdot \|\mu(A)-c\|^2$.
        \item $\sum_{p,q\in A} \|p-q\|^2 = 2\cdot |A|\cdot \sum_{p\in A} \|p-\mu(A)\|^2$.
    \end{itemize}
\end{lemma}
We note that as an immediate corollary, the lemma implies that the sum of squared distances of all points in a cluster $C_i$ to an approximate mean is at most $(1+\varepsilon)|C_i|\Delta_i$ and the MinSum clustering cost is at most $(1+\varepsilon)|C_i|^2\Delta_i$.
\begin{corollary}
\label{cor:approxmean}
For any set of points $A\subset \mathbb{R}^d$. Then $c\in \mathbb{R}^d$ is an $\varepsilon$-approximate mean of $A$ if and only if $\sum_{p\in A} \|p-c\|^2 \leq (1+\varepsilon)\cdot |C_i|\cdot \Delta_i$.
\end{corollary}
\begin{lemma}
    \cite{BecchettiBC0S19}
    \label{lem:weaktriangle}
    Given numbers $a,b,c$, we have for all $\varepsilon>0$
    $$(a-b)^2 \leq (1+\varepsilon)\cdot (a-c)^2 + \left(1+\frac{1}{\varepsilon}\right)(b-c)^2.$$
\end{lemma}

We also show that we only have to consider seeding sets with $ \beta \in \Theta(\varepsilon^{-2})$.
\begin{lemma}
\label{lem:magiccore}
For any cluster $C_i$, $\varepsilon\in (0,1)$ and $\beta\geq 12\varepsilon^{-2}$, we have that $\mu_i(C_i^{\beta}) = \frac{1}{|C_i^{\beta}|} \sum_{p\in C_i^{\beta}}p$ is a $\varepsilon$-approximate mean of $C_i$.
\end{lemma}
\begin{proof}
By Markov's inequality, $|C_i\setminus C_i^{\beta}|\leq \beta^{-1} \cdot|C_i|\leq \frac{\varepsilon^2}{12}\cdot |C_i|$.
Since $\frac{1}{|C_i^{\beta}|}\sum_{p\in C_i^{\beta}}\|p-\mu_i(C_i^{\beta})\|^2 \leq \frac{1}{|C_i^{\beta}|}\sum_{p\in C_i^{\beta}}\|p-\mu_i\|^2 \leq \Delta_i\cdot \frac{|C_i|}{|C_i^{\beta}|}\leq 2\Delta_i$, Lemma \ref{lem:magic} implies $\|\mu_i(C_i^{\beta})-\mu_i\|^2 = \frac{1}{|C_i^{\beta}|}\sum_{p\in C_i^{\beta}}\|p-\mu_i\|^2 - \frac{1}{|C_i^{\beta}|}\sum_{p\in C_i^{\beta}}\|p-\mu_i(C_i^\beta)\|^2 \leq 2\Delta_i$.
We then have due to Lemma \ref{lem:weaktriangle}
\begin{align*}
\sum_{q\in C_i\setminus C_i^{\beta}} \|q-\mu_i(C_i^{\beta})\|^2 - \|q-\mu_i\|^2 & \leq \frac{\varepsilon}{2} \sum_{q\in C_i\setminus C_i^{\beta}} \|p-\mu_i\|^2 + \left(1+\frac{2}{\varepsilon} \right)\cdot  \|\mu_i - \mu_i(C_i^\beta)\|^2 \\
&\leq \frac{\varepsilon}{2} \sum_{q\in C_i\setminus C_i^{\beta}} \|p-\mu_i\|^2 + \frac{\varepsilon^2}{12} |C_i|\cdot \frac{3}{\varepsilon} \cdot 2\Delta_i \leq \varepsilon |C_i|\Delta_i
\end{align*}
The cost of the points in $C_i^\beta$ to $\mu_i(C_i^\beta)$ only gets smaller compared to the cost of these points to $\mu_i$. Hence, the increase in cost is bounded by $\varepsilon |C_i|\Delta_i$, which with Corollary \ref{cor:approxmean} yields the claim.
\end{proof}

Finally, we also show how to efficiently extract a mean from a mean seeding set, while being oblivious to $\Delta_i$.
\begin{lemma}
\label{lem:constructive}
    Let $S$ be an $(\varepsilon/4,\beta)$-mean seeding set of a cluster $C_j$ with mean $\mu_j$. Then we can compute    $\left(\frac{10\beta\cdot |S|}{\varepsilon}+1\right)^{|S|}$ choices of weights in time linear in the size of choices such that at least one of the computed choices satisfies
    $$\left\|\frac{1}{\sum_{p\in S} w(p)}\sum_{p\in S} w(p) \cdot p - \mu_j\right\|^2\leq \varepsilon\cdot \Delta_j.$$
\end{lemma}
\begin{proof}
We first introduce some preprocessing. By an affine transformation of the space, subtract $q = \underset{p\in S}{\text{argmin}}\|p-\mu_j\|$ from all points. Now all points $p$ in $S$ with $\|p-\mu_j\|\leq \sqrt{\beta\cdot \Delta_j}$ have norm at most $2\sqrt{\beta\cdot \Delta_j}$. 

Let $S'\subset S$ be the set with weights $w$ such that
$$\left\|\frac{1}{\sum_{p\in S'}w(p)}\sum_{p\in S'}w(p)p - \mu_j\right\|^2\leq \frac{\varepsilon}{4}\cdot \Delta_j.$$

Let $w_{\max}$ be the maximum weight of the points in $S'$. For $w(p)$ every $p$, we set $w'(p)$ to be the largest multiple of $\frac{\varepsilon}{10\beta\cdot |S|}\cdot w_{\max}$ that is
at most $w(p)$ (where we extend $w$ to all of $S$ by setting $w(p)=0$ for all $p\not\in S'$). So,
$w'(p) = \frac{\varepsilon}{10\beta\cdot |S|}\cdot i \cdot w_{\max} \leq w(p)< \frac{\varepsilon}{10\beta\cdot |S|}\cdot (i+1) \cdot w_{\max} $ for some $i\in \left\{0,1,\ldots \frac{10\beta\cdot |S|}{\varepsilon}\right\}$. Observe that there are at most $\left(\frac{10\beta\cdot |S|}{\varepsilon}+1\right)^{|S|}$
choices of weights of points in $S$.
Furthermore, we have
$$
\left\vert\sum_{p\in S}\left(w'(p)-w(p)\right)\right\vert \leq \frac{\varepsilon}{10\beta\cdot |S|}\sum_{p\in S'}w(p).
$$
We now argue that $\mu' = \frac{1}{\sum_{p\in S'}w'(p)}\sum_{p\in S'} w'(p)$ is a $\varepsilon$-approximate mean of $C_j$.
We have
\begin{eqnarray*}
    & & \left\|\frac{1}{\sum_{p\in S'}w(p)}\sum_{p\in S'}w(p)p - \frac{1}{\sum_{p\in S'}w'(p)}\sum_{p\in S'}w'(p)p\right\| \\
    &=& \frac{1}{\sum_{p\in S'}w'(p)}\left\|\frac{\sum_{p\in S'}w'(p)}{\sum_{p\in S'}w(p)}\sum_{p\in S'}w(p)p - \sum_{p\in S'}w'(p)p\right\| \\
    &=& \frac{1}{\sum_{p\in S'}w'(p)}\left\|\frac{\sum_{p\in S'}w'(p)-\sum_{p\in S'}w(p)}{\sum_{p\in S'}w(p)}\sum_{p\in S'}w(p)p\right\| + \frac{1}{\sum_{p\in S'}w(p)}\left\|\sum_{p\in S'}(w(p)-w'(p))p\right\|\\
    &\leq & \frac{2\varepsilon}{10\beta\cdot |S|}\left\|\frac{1}{\sum_{p\in S'}w(p)}\sum_{p\in 'S}w(p)p\right\| + \sum_{p\in S'} \frac{\varepsilon}{10\beta\cdot |S|}\|p\| \\
    &\leq & \frac{5\varepsilon}{10\beta\cdot |S|} \cdot |S'| \cdot \sqrt{\beta\cdot \Delta_j} \leq \frac{\varepsilon}{2} \sqrt{\Delta_j}
\end{eqnarray*}
By the triangle inequality, we can therefore conclude that $\mu'$ is a $\varepsilon$-approximate mean of $\mu_j$
\end{proof}

\paragraph{Computing a Mean-Seeding Set via Uniform Sampling.}

\begin{lemma}
\label{lem:uniform}
Let $\varepsilon\in (0,1)$ and $\beta > 48\varepsilon^2$.
With probability at least $1-\delta$, a set of $32 k\varepsilon^{-1}\log \delta^{-1}$ points $S$ sampled uniformly at random with replacement from $A$ contains is a $(\varepsilon,\beta)$-mean seeding set of any $C_i$ with $|C_i|\geq \frac{n}{k}$.
\end{lemma}
\begin{proof}
Due to Lemma \ref{lem:magiccore}, The mean of $C_i^{\beta}$ is an $\frac{\varepsilon}{2}$-approximate mean. Hence, if we obtain a $(\varepsilon/2,\beta)$-seeding set of $C_i^{\beta}$, the claim follows. By Markov's inequality, $C_i^{\beta}$ contains at least $\frac{n}{2k}$ points.
For any $p\in C_i^{\beta}$, we have $\mathbb{E}\left[\|p-\mu(C_i^{\beta})\|^2\right] = \Delta_i^{\beta} := \frac{1}{|C_i^{\beta}|}\sum_{p\in C_i^{\beta}}\|p-\mu(C_i^{\beta}\|^2\leq \Delta_i$ and therefore for any set of $m$ points $S_i$ sampled independently with replacement from $C_i^{\beta}$, $\mathbb{E}\left[\|\frac{1}{m}\sum_{p\in S_i} p-\mu(C_i^{\beta})\|^2\right] = \frac{1}{m}\Delta_i^{\beta}$. Therefore, if $m\geq 4\varepsilon^{-1}$, $S_i$ is an $(\varepsilon/2,\beta)$-mean seeding set of $C_i$ with probability at least $\frac 1 2$. Hence, sampling $\log\delta^{-1}$ many copies of $S_i$ implies that at least one of them is an $(\varepsilon/2,\beta)$-mean seeding set of $C_i$ with probability $1-2^{-\log \delta^{-1}} = 1-\delta$.

A sample from the point set is contained in $C_i^{\beta}$ with probability at least $\frac{1}{2k}$. Hence, sampling at least $16k\cdot \varepsilon^{-1}\cdot \log \delta^{-1}$ implies that with probability at least $1-\delta$, the number $X$ of points sampled from $S_i$ is
at least $4\epsilon^{-1}\log\delta^{-1}$, as follows. By the above analysis $\E X \ge 8\eps^{-1}\log\delta^{-1}$. Therefore, by standard
Chernoff bounds, $\Pr[X < 4\epsilon^{-1}\log\delta^{-1}] < e^{-\frac 1 8\cdot 8\eps^{-1}\log\delta^{-1}} \le \delta$. 
\end{proof}



\paragraph{$D^2$ Subsampling}
We now define an algorithm for sampling points that induce means from the target clusters.
The high level idea is as follows. We construct a rooted tree in which every node is labeled by a set of candidate cluster means. For a parent and child pair of nodes, the parent's set is a subset of the child's set. The construction is iterative. Given an interior node, we construct its children by adding a candidate mean to the parent's set. The candidantes are generated using points sampled at random from a distribution that will be defined later. The goal is to have, eventually, an $\varepsilon$-approximate mean for every optimal cluster. This will be achieved with high probability at one of the leaves of the tree. The root of the tree is labeled with the empty set, and its children are constructed via uniform sampling. Subsequently, we refine the sampling distribution to account for various costs and densities of the clusters.

We now go into more detail for the various sampling stages of the algorithm.

\begin{description}
    \item[Preprocessing:] We ensure that all points are not too far from each other.
    \item[Initialization:] We initialize the set of means via uniform sampling. Due to Lemma \ref{lem:uniform}, we can enumerate over potential sets of $\varepsilon$-approximate means for all clusters of size $\frac n k$. Each candidate mean defines a child of the root.
    \item[Sampling Stage:] Consider a node of the tree labeled with a non-empty set of candidate means $M$. We put 
    $\Gamma_i = 2^{-i}\cdot\sum_{q\in A}\min_{m\in M}\|q-m\|^2$ for $i\in \{0,1,\ldots,13\log ( nk / \varepsilon)\}$, where $\eta$ is an absolute constant to be defined later. Let $A_{i,M} = \{q\in A\colon \min_{m\in M}\|q-m\|^2\le \Gamma_i\}$. (Note that $A_{0,M}$ includes all the points.)
    Let $\mathbb{P}_i$ denote the probability distribution on $A_{i,M}$ induced by setting, for each $p\in A_{i,M}$, 
    $$
    \mathbb{P}_i[p]=\frac{\min_{m\in M}\|p-m\|^2}{\sum_{p\in A_{i,M}}\min_{m\in M}\|p-m\|^2}
    $$ 
    We'll use $\mathbb{P}$ to denote $\mathbb{P}_0$.
    For each $i$, we sample a sufficient (polynomial in $k$ and $\varepsilon$, but independent of $n$) number of points independently from the distribution $\mathbb{P}_i$. Let $S$ denote the set of sampled points.
    \item[Mean Extraction Stage:] We enumerate over combinations of points in $M\cup S$, using some non-uniform weighing to fix a mean to add to $M$, see Lemma~\ref{lem:constructive}. Each choice of mean is added to $M$ to create a child of the node labeled $M$.
\end{description}

Throughout this section we will use the following definition. Given a set of centers $M$, we say that a cluster $C_i$ is $\varepsilon$-covered by $M$ if $|C_i|^2\cdot \min_{m\in M}\|\mu_i-m\|^2\leq \frac{\varepsilon}{2}\cdot \left(\frac{1}{k}\cdot OPT + |C_i|^2 \Delta_i\right)$.
Our goal will be to prove the following lemma.
\begin{lemma}
\label{lem:approxmeans}
    Let $\mathcal{C}=\{C_1,\ldots C_k\}$ be the clusters of an optimal Min-Sum $k$-clustering and let $\eta$ be an absolute constant.
    For every $\delta, \epsilon > 0$, there is a randomized algorithm that 
    outputs a collection of at most $n^{o(1)} \cdot 2^{\eta\cdot k^2 \cdot \varepsilon^{-12} \log^2(k/(\varepsilon\delta))}$ sets of at most $k$ centers $M$, such that with probability $1 -\delta$ at least one of them that $\varepsilon$-covers every $C_i \in \mathcal{C}$. The algorithm runs in time $n^{1+o(1)} \cdot d \cdot 2^{\eta\cdot k^2 \cdot \varepsilon^{-12} \log^2(k/(\varepsilon\delta))}$.
\end{lemma}

Note that if all clusters of $\mathcal{C}$ are $\varepsilon$-covered, then there exists an assignment of points to centers, such that Min-Sum clustering cost of the resulting clustering is at most $(1+\varepsilon)\cdot \OPT$. To see this, notice that if we use $\mathcal{C}$ as the clustering with $m_i=\text{argmin}_{m\in M}\|\mu_i-m\|^2$, then 
$$\sum_{i=1}^k |C_i| \sum_{p\in} \|p-m_i\|^2 \leq \OPT + \sum_{i=1}^k |C_i| \sum_{p\in} \frac{\varepsilon}{2} \left(\frac{1}{k}\cdot \frac{\OPT}{|C_i|^2} + \frac{1}{2}\Delta_i\right) \leq (1+\varepsilon)\cdot \OPT.$$





\paragraph{Preprocessing}

The first lemma allows us to assume that all points are in some sense close to each other.
\begin{lemma}
\label{lem:preprocess}
    Suppose $n>20$.
    Given an set of $n$ points $A\subset \mathbb{R}^d$, we can partition a point set into subsets $A_1,\ldots A_k$, such that $\|p-q\|^2 \leq n^{10}\cdot \OPT$ for any two points $p,q\in A_i$ and such that any cluster $C_j$ is fully contained in one of the $A_i$. The partitioning takes time $\tilde{O}(nd+k^2)$.
\end{lemma}
\begin{proof}
The proof uses similar arguments found throughout $k$-means and $k$-median research, with only difference being that some of the discretization arguments are slightly finer to account for the MinSum clustering objective.

Consider a candidate $20$-approximate $k$-means clustering with cost $T$, which can be computed in time $\tilde{O}(nd+k^2)$~\cite{DraganovSS24}. Then we have $\frac{1}{20}T\cdot \OPT\leq 20 n^2 \cdot T$. Now, suppose that there are two centers $c_1$ and $c_2$ such that $\|c_1-c_2\|^2 \leq 20 n^{7} \cdot T$. Then for any point $p\in C_1$ and $q\in C_2$, we have by the triangle inequality $\|p-q\|^2 \leq 20n^{9}\cdot T\leq n^{10}\cdot T$. Conversely, if $\|c_1-c_2\|^2 > n^{8} \cdot T$, we know that no two points in the clusters induced by $C_1$ and $C_2$ can be in the same cluster of the optimal MinSum clustering. 
\end{proof}

\paragraph{Computing a Mean-Seeding Set via $D^2$ Sampling.}

We now consider a slight modification of Lemma \ref{lem:uniform} to account for sampling points from a cluster non-uniformly. We introduce the notion of a distorted core as follows.
Given a cluster $C_j$, a set of centers $M$, and parameters $\alpha,\beta$, we say that a subset of $C_j^{\beta}\cup M$ is a $(C_j,\beta,\alpha,M)$-distorted core (denoted $core(C_j,\beta,\alpha,M)$) iff it is the image of a mapping $\pi_{\alpha,M}:C_j^\beta\rightarrow C_j^\beta\cup M$ such that for any point $p\in C_j^\beta$, we have
$$\pi_{\alpha,M}(p) = \begin{cases} p &\text{if }\min_{m\in M}\|p-m\|^2 \geq \alpha\cdot \Delta_j \\  \underset{m\in M}{\text{argmin}}\|p-m\|^2 & \text{if }\min_{m\in M}\|p-m\|^2 < \alpha\cdot \Delta_j\end{cases}.$$
We use $D(C_j,\beta,\alpha,M)$ to denote the set of points in $C_j^\beta$ such that $\min_{m\in M}\|p-m\|^2 < \alpha\cdot \Delta_j$.

The following lemmas relate the goodness of a mean computed on an $\alpha$-distorted core to the mean on the entire set of points when sampling points proportionate to squared distances. We start by proving an analogue of Lemma \ref{lem:magiccore}.
\begin{lemma}
\label{lem:core}
Let $\alpha \leq\frac{\varepsilon}{4}$ and let $\beta\geq \frac{144}{\varepsilon^{2}}$.
Given a set of centers $M$ and a cluster $C_j$, let 
$$\hat{\mu_j} = \frac{1}{|C_j^\beta|}\sum_{p\in C_j^\beta} \pi_{\alpha,M}(p).$$ 
Then,
$$\|\hat{\mu_j}-\mu_j\|^2\leq \varepsilon\cdot \Delta_j.$$
\end{lemma}
\begin{proof}
First, let $\mu_j'$ be the mean of $C_j^\beta$. Due to Markov's inequality $|C_j^\beta|\geq \frac{|C_j|}{2}$. Using Lemma \ref{lem:magic}, we have $|C_j|\cdot \Delta_j \geq \sum_{p\in C_j^\beta}\|p-\mu_j\|^2 \geq |C_j^\beta|\cdot\|\mu_j'-\mu_j\|^2$, which implies that $\|\mu_j'-\mu_j\|^2 \cdot |C_j|\leq 2|C_j| \cdot \Delta_j$. Then
\begin{eqnarray*}
 \sum_{p\in C_j} \|p-\mu_j'\|^2 
& = & \sum_{p\in C_j^\beta}\|p-\mu_j'\|^2 + \sum_{p\in C_j\setminus C_j^\beta}\|p-\mu_j'\|^2\\
&\leq& \sum_{p\in C_j^\beta}\|p-\mu_j\|^2 + \\
& & + \sum_{p\in C_j\setminus C_j^\beta}\left(1+\frac{\varepsilon}{8}\right)\cdot \|p-\mu_j\|^2 + |C_j\setminus C_j^\beta|\cdot \left(1+\frac{8}{\varepsilon}\right)\cdot \|\mu_j'-\mu_j\|^2 \\
&\leq& \left(1+\frac{\varepsilon}{8}\right)\cdot \sum_{p\in C_j}\|p-\mu_j\|^2  + \frac{9}{\varepsilon\beta} \cdot |C_j| \cdot \|\mu_j'-\mu_j\|^2 \\
&\leq& \left(1+\frac{\varepsilon}{8}\right)\cdot \sum_{p\in C_j}\|p-\mu_j\|^2  + \frac{18}{\varepsilon\beta} \cdot |C_j|\Delta_j,
\end{eqnarray*}
where we used Lemma \ref{lem:weaktriangle} in the second inequality.
In other words, $\mu_j'$ is an $\left(\frac{\varepsilon}{8} + \frac{18}{\varepsilon\beta}\right)$-approximate mean of $C_j$. We now turn our attention to $\hat{\mu_j}$. We have
\begin{align*}
\|\hat{\mu_j}-\mu_j\| & \leq \frac{1}{|C_j^\beta|}\cdot \sum_{p\in C_j^\alpha}\|p-\pi_{\alpha,M}(p)\| \leq \sqrt{\alpha\cdot \Delta_j}
\end{align*}
By the triangle inequality, we therefore have
$$\|\hat{\mu_j}-\mu_j\|\leq \|\hat{\mu_j} - \mu_j'\|+\|\mu_j'-\mu_j\| \leq \sqrt{\alpha\cdot \Delta_j} + \sqrt{\left(\frac{\varepsilon}{8} + \frac{18}{\varepsilon\beta}\right)\cdot \Delta_j}.$$
By our choice of $\alpha$ and $\beta$, this implies that $\hat{\mu_j}$ is an $\varepsilon$-approximate mean of $C_j$. 
\end{proof}

We now characterize when $M$ either covers a cluster $C_j$, or when $M$ is a suitable seeding set for $C_j$. The following lemma says that if $M$ is not a seeding set of $C_j$, then there exist many points in the core $C_j^\beta$ of $C_j$ that are far from $M$.

\begin{lemma}
\label{lem:closecore}
    Given $\alpha\leq \frac{\varepsilon}{16}$, $\beta \geq \frac{2400}{\varepsilon^{2}}$, and $\gamma\leq \sqrt{\frac{\varepsilon}{16(\beta+\alpha)}}$,
    and a set of centers $M$, let $C_j$ be a cluster for which $|D(C_j,\beta,\alpha,M)| \geq (1-\gamma)\cdot |C_j^\beta|$. 
    Then $M$ is an $(\varepsilon,\beta)$-mean seeding set of $C_j$. 
\end{lemma}
\begin{proof}
    First, let $\hat{\mu_j} = \frac{1}{|C_j^\beta|}\sum_{p\in C_j^\beta}\pi_{\alpha,M}(p)$ and let 
    $\mu_j' = \frac{1}{|D(C_j,\beta,\alpha,M)|}\sum_{p\in D(C_j,\beta,\alpha,M)}p$ be the mean of $D(C_j,\beta,\alpha,M)$. 
    Now, observe that for any pairs of points $p\in C_j^\alpha$ and $q\in C_j^\beta$, by the triangle inequality 
    $$
    \|q-p\|\leq \|q - \mu_j\| + \|\mu_j - p\|\le \sqrt{(\beta+\alpha)\cdot \Delta_j}.
    $$
    Then
    \begin{eqnarray*}
        & & \|\hat{\mu_j}- \mu_j'\| \\
        & = & \frac{1}{|C_j^\beta|} \left\|\sum_{p\in C_j^\beta}\pi_{\alpha,M}(p) - \frac{|C_j^\beta|}{|D(C_j,\beta,\alpha,M)|}\sum_{p\in D(C_j,\beta,\alpha,M)}p\right\|\\
        & = &\frac{1}{|C_j^\beta|}\cdot \left\|\left(\sum_{p\in D(C_j,\beta,\alpha,M)}(\pi_{\alpha,M}(p)-p)\right)\right. + \\
        & & + \left.\left(\sum_{p\in C_j^\beta\setminus D(C_j,\beta,\alpha,M)}\pi_{\alpha,M}(p) - \frac{|C_j^\beta\setminus D(C_j,\beta,\alpha,M)|}{|D(C_j,\beta,\alpha,M)|}\sum_{p\in D(C_j,\beta,\alpha,M)}p\right)\right\|\\
        &\leq& \frac{1}{|C_j^\beta|}\cdot \left\|\sum_{p\in D(C_j,\beta,\alpha,M)}(\pi_{\alpha,M}(p)-p)\right\| + \\
        & & + \frac{1}{|C_j^\beta|}\cdot
        \left\|\sum_{p\in C_j^\beta\setminus D(C_j,\beta,\alpha,M)}\pi_{\alpha,M}(p) - \frac{|C_j^\beta\setminus D(C_j,\beta,\alpha,M)|}{|D(C_j,\beta,\alpha,M)|}\sum_{p\in D(C_j,\beta,\alpha,M)}p\right\| \\
        &\leq& \sqrt{\alpha \cdot \Delta_j} + \gamma\cdot \sqrt{(\beta+\alpha)\cdot \Delta_j}
    \end{eqnarray*}
    Finally, by the triangle inequality, Lemma \ref{lem:core} and our choice of $\alpha$, $\beta$, and $\gamma$, we have
    $$\|\mu_j'-\mu_j\|\leq \|\mu_j'-\hat{\mu_j}\| + \|\hat{\mu_j}-\mu_j\|\leq \sqrt{\alpha \cdot \Delta_j} + \gamma\cdot \sqrt{(\beta+\alpha)\cdot \Delta_j} + \sqrt{\frac{\varepsilon}{4}\Delta_j} \le \sqrt{\epsilon\Delta_j},$$
thus completing the proof.
\end{proof}

As a consequence of this lemma and the preprocessing, we show under the assumption of Lemma \ref{lem:preprocess}, the largest value of $i$ such that $C_j^\beta \in A_{i,M}$ for an uncovered cluster $C_j$ cannot be too large.
\begin{lemma}
\label{lem:preprocess2}
    Given $\beta \geq \frac{2400}{\varepsilon^{2}}$,
    suppose we have a set of points $A$ such that $\|p-q\|^2 \leq n^{10}\cdot \OPT$ as per Lemma \ref{lem:preprocess}. 
    Let $M$ be a set of points and suppose there exists a cluster $C_j$ such that such $C_j$ is uncovered and such that $M$ is not an $(\varepsilon/4,\beta)$ mean seeding set of $A$. We then have that $C_j^\beta \subset A_{i,M}$ implies $i\leq 13 \log(nk/\varepsilon)$.
\end{lemma}
\begin{proof}
    Suppose $i> 13 \log(nk/\varepsilon)$. 
    Due to Lemma \ref{lem:closecore}, we know there exists a point $p'\in C_j^\beta$ such that $\min_{m\in M} \|p-m\|^2 \geq \varepsilon/16\cdot \Delta_j$. This implies via Lemma \ref{lem:preprocess} that $\Delta_j \leq \left(\frac{k \cdot n}{\varepsilon}\right)^{-13}\cdot 16\varepsilon^{-1}\cdot n^{10}\cdot\OPT$.
    
    Consider the point $p\in C_j^\beta$ with minimumal distance to $\mu_j$ and let $m_p=\text{argmin}_{m\in M} \|p-m\|^2$. Then $\|p-m\|^2 \leq n^{-20}\cdot \OPT$, which implies that 
    \begin{align*}
        |C_j|\cdot \sum_{q\in C_j}\|q-m\|^2 &\leq |C_j|\cdot\sum_{q\in C_j}2\cdot\|q-p\|^2+2\cdot\|p-m\|^2 \\
        & \leq  4|C_j|\cdot \sum_{q\in C_j}\|q-\mu_j\|^2+2 |C_j|^2\cdot\|p-m\|^2 \\
        &\leq 4|C_j|^2\cdot 16 \varepsilon^{-1}\cdot \left(\frac{k \cdot n}{\varepsilon}\right)^{-13}\cdot n^{10}\cdot\OPT + 2 |C_j|^2 \left(\frac{k \cdot n}{\varepsilon}\right)^{-13}\cdot  n^{10}\cdot\OPT \\
        &\leq 66\cdot |C_j|^2\cdot\left(\frac{k \cdot n}{\varepsilon}\right)^{-13}\cdot 16\varepsilon^{-1}\cdot n^{10}\cdot\OPT \leq \frac{\varepsilon}{2k}\cdot\OPT,
    \end{align*}
     which is a contradiction to $M$ not covering $C_j$.
\end{proof}

We now show that, given that $M$ is not a seeding set of some cluster $C_j$, that the weighted squared distance of $\mu_j$ to its closest point in $M$ is a reasonably accurate proxy for the squared distance of the points in the core $C_j^\beta$ to their respectively closest points in $M$.

\begin{lemma}
\label{lem:lbdist}
Let $M$ be a set of centers and let $C_j$ be a cluster that is not $\varepsilon$-covered by $M$. Also assume that $M$ is 
not an $(\varepsilon,\beta)$-mean seeding set of $C_j$. Then,
$$\sum_{p\in C_j^\beta} \min_{m\in M}\|p-m\|^2 \geq \frac{1}{272}\left(\frac{\varepsilon}{\beta}\right)^{3/2}\cdot |C_j| \cdot \min_{m\in M}\|\mu_j-m\|^2$$
\end{lemma}
\begin{proof}
For all $p\in C_j^\beta\setminus D(C_j,\beta,\varepsilon/16,M)$, we have:
\begin{eqnarray*}
       \min_{m\in M}\|\mu_j-m\|^2 
&\leq& \left(\min_{m\in M}\|p-m\| + \|\mu_j-p\|\right)^2 \\
&\leq& 2\min_{m\in M}\|p-m\|^2 + 2\|\mu_j-p\|^2 \\
&\leq& \frac{34\beta}{\varepsilon}\cdot\min_{m\in M}\|p-m\|^2,
\end{eqnarray*}
where the first inequality uses the triangle inequality and that for $m' = \arg\min_{m\in M}\|p-m\|$, we have that $\|\mu_j-m'\|^2\ge \min_{m\in M}\|\mu_j-m\|^2$, and the last inequality uses $\min_{m\in M}\|p-m\|^2\ge \frac{\varepsilon}{16}\Delta_j$ 
and $\|\mu_j-p\|^2\le \beta\Delta_j$ and $\frac{\beta}{\varepsilon} \ge 1$.

We first consider the case that $\min_{m\in M}\|\mu_j-m\| \geq 2\sqrt{\beta\cdot \Delta_j}$. In this case, all points in $C_j^\beta$ are closer to $\mu_j$ than to any point in $M$.
This implies 
$$\sum_{p\in C_j^\beta} \min_{m\in M}\|p-m\|^2 \geq \frac{1}{4}|C_j^\beta| \min_{m\in M}\|\mu_j-m\|^2 \geq \frac{1}{8}|C_j|\min_{m\in M}\|\mu_j-m\|^2.$$

Now, we consider the case that $\min_{m\in M}\|\mu_j-m\| \leq 2\sqrt{\beta\cdot \Delta_j}$.
As $M$ is not an $\varepsilon$-mean seeding set for $C_j$, Lemma~\ref{lem:closecore} implies that
$|C_j^\beta\setminus D(C_j,\beta,\varepsilon/16,M)| > \sqrt{\frac{\varepsilon}{16\beta+\varepsilon}} |C_j^\beta|\ge \frac 1 2 \sqrt{\frac{\varepsilon}{16\beta+\varepsilon}} |C_j|$. Therefore,
\begin{eqnarray*}
      \sum_{p\in C_j^\beta} \min_{m\in M}\|p-m\|^2
&\ge& \sum_{p\in C_j^\beta\setminus D(C_j,\beta,\varepsilon/16,M)} \min_{m\in M}\|p-m\|^2 \\
&\ge& |C_j^\beta\setminus D(C_j,\beta,\varepsilon/16,M)|\cdot \frac{\varepsilon}{34\beta}\cdot\min_{m\in M}\|\mu_j-m\|^2 \\
&\ge& \frac 1 2 \sqrt{\frac{\varepsilon}{16\beta + \varepsilon}}\cdot\frac{\varepsilon}{34\beta}\cdot |C_j|\cdot\min_{m\in M}\|\mu_j-m\|^2 \\
&\ge& \frac{1}{272}\left(\frac{\varepsilon}{\beta}\right)^{3/2}\cdot |C_j|\cdot\min_{m\in M}\|\mu_j-m\|^2,
\end{eqnarray*}
which completes the proof.
\end{proof}

Next, we show that the marginal probability of picking a point from an uncovered cluster $C_j$ cannot be significantly smaller than the marginal probability of picking a point from the union of covered clusters with larger cardinality than $C_j$.

\begin{lemma}
\label{lem:lbprob}
    Let $M$ be a set of centers, and let ${\cal C}$ denote a set of clusters that are $\varepsilon$-covered by $M$. Let $\mathcal{H}$ denote the set of points in all the clusters in ${\cal C}$. Let $\beta> \frac{2400}{\varepsilon^{2}}$. Consider a cluster $C_j\not\in {\cal C}$. Let $i$ be the largest index such that $C_i\in {\cal C}$. Suppose that $M$ is not an $(\varepsilon,\beta)$-mean seeding set of $C_j$, and that $i < j$. Then
    $$\mathbb{P}[p\in C_j^\beta\mid p\in \mathcal{H}\cup C_j] \geq \frac{\varepsilon^4 \cdot \beta^{-3/2}}{1088k}.$$
\end{lemma}
\begin{proof}
For the points in $\mathcal{H}\cup C_j$, we have 
\begin{align*}
\label{eq:sample1}
  \sum_{p\in \mathcal{H}\cup C_j}\min_{m\in M}\|p-m\|^2 & = \sum_{C_h\in \cal C}  |C_h| \cdot \left(\Delta_h + \min_{m\in M}\|\mu_h - m\|^2\right) + |C_j|\cdot(\Delta_j + \min_{m\in M}\|\mu_j-m\|^2) \\
  &\leq \sum_{C_h\in \cal C} (1+\varepsilon) \cdot |C_h|\cdot\Delta_h + |C_j|\cdot(\Delta_j + \min_{m\in M}\|\mu_j-m\|^2) \\
  &\leq  2\cdot \left(\sum_{C_h\in \cal C} |C_h|\cdot\Delta_h  +\varepsilon^{-1}\cdot |C_j| \min_{m\in M}\|\mu_j-m\|^2\right),
\end{align*}
where the first inequality holds by definition of an $\varepsilon$-covered cluster and the second inequality holds as $M$ does not $\varepsilon$-cover $C_j$ and thus in particular $\min_{m\in M}\|\mu_j-m\|^2 \geq \varepsilon\cdot \Delta_j$ due to Corollary \ref{cor:approxmean}.

Assume for contradiction that the lemma does not hold, so
$$
\sum_{p\in C_j^\beta}\|p-m\|^2 <  \frac{\varepsilon^4 \cdot \beta^{-3/2}}{1088k}\cdot \sum_{p\in \mathcal{H}\cup C_j}\min_{m\in M}\|p-m\|^2.
$$
This yields
\begin{eqnarray*}
\frac{1}{272}\left(\frac{\varepsilon}{\beta}\right)^{3/2} |C_j| \cdot \min_{m\in M}\|\mu_j-m\|^2
&\le& \sum_{p\in C_j^\beta}\|p-m\|^2 \\
& < & \frac{\varepsilon^4 \cdot \beta^{-3/2}}{1088k}\cdot \sum_{p\in \mathcal{H}\cup C_j}\min_{m\in M}\|p-m\|^2\\
&\le&  \frac{\varepsilon^4 \cdot \beta^{-3/2}}{544 k}\cdot \left(\sum_{C_h\in \cal C}  |C_h|\cdot\Delta_h  +\varepsilon^{-1}\cdot |Cj| \min_{m\in M}\|\mu_j-m\|^2\right), \\
\end{eqnarray*}
where the first inequality uses Lemma~\ref{lem:lbdist}. (Note that this lemma assumes that $M$
is not an $(\varepsilon,\beta)$-seeding set for $C_j$.) Rearranging the terms, we get
\begin{eqnarray*}
\frac{\varepsilon^3\cdot \beta^{-3/2}\cdot}{544}\cdot |C_j|\cdot\min_{m\in M}\|\mu_j-m\|^2
&\le&\left(\frac{(\varepsilon/\beta)^{3/2}}{272}-\frac{(\varepsilon/\sqrt{\beta})^{3}}{544k}\right)\cdot|C_j|\cdot\min_{m\in M}\|\mu_j-m\|^2 \\
&\leq& \frac{\varepsilon^4 \cdot \beta^{-3/2}}{544k}\cdot \sum_{C_h\in \cal C}\cdot  |C_h|\cdot\Delta_h.
\end{eqnarray*}
Therefore, as $|C_j|\leq |C_h|$ for all $C_h\in {\cal C}$,
$$
|C_j|^2 \min_{m\in M}\|\mu_j-m\|^2 \leq \frac{\varepsilon}{k} \sum_{C_h\in \cal C}  |C_h|\cdot\Delta_h \cdot |C_j| \leq \frac{\varepsilon}{k}\cdot |C_h|^2\cdot\Delta_h.
$$
This, however, implies that $C_j$ is $\varepsilon$-covered by $M$, contradicting the lemma's assumption.
\end{proof}

We now consider a cluster $C_j$ that is small compared to the union of the clusters $C_j'$ with $j'>j$. In this case, we show that one of the distance-proportional distributions that we use guarantees that the probability of sampling points from the core of $C_j$ is large.

\begin{lemma}
\label{lem:refinedsample}
 Let $M$ be a set of centers. Let $\beta> \frac{2400}{\varepsilon^{2}}$. Let $j$ be the smallest index 
 such that $C_j$ is not $\varepsilon$-covered by $M$. If $M$ is not an $(\varepsilon,\beta)$-mean seeding
 set for $C_j$, then there exists $i\in \{0,1,\ldots,\eta\log (nk/\varepsilon)\}$ such that 
 $C_j^\beta \in A_{i,M}$ and
$$\mathbb{P}_i[p\in C_j^\beta]\geq \frac{1}{4352\cdot k}\cdot \left(\frac{\varepsilon}{\beta^{5/8}}\right)^4$$
\end{lemma}
\begin{proof}
By Markov's inequality $|C_j|/2 < |C_j^\beta|$.
Let $i$ be the smallest value such that $C_j^{\beta} \subset A_{i,M}$. (Clearly, $C_j^{\beta} \subset A_{0,M}$, so $i$ exists.)
We have due to Lemma \ref{lem:lbdist}
$$\sum_{p\in C_j^\beta} \min_{m\in M}\|p-m\|^2 \geq \frac{1}{272}\left(\frac{\varepsilon}{\beta}\right)^{3/2}\cdot |C_j| \cdot \min_{m\in M}\|\mu_j-m\|^2,$$
Also, for all $p\in C_j^\beta$,
$$ \min_{m\in M}\|p-m\| \leq \min_{m\in M}\|\mu_j-m\| + \|p-\mu_j\| \leq \min_{m\in M}\|\mu_j-m\| + \sqrt{\beta \cdot \Delta_j} < 2\sqrt{\frac{\beta}{\varepsilon}} \min_{m\in M}\|\mu_j-m\|$$
where the last inequality follows from the fact that $M$ does not $\varepsilon$-cover $C_j$, so 
$\min_{m\in M}\|\mu_j-m\|^2 > \varepsilon\cdot \Delta_j$
Note that this implies $\min_{m\in M}\|p - m\|^2 \leq 8\cdot \frac{\beta}{\varepsilon} \cdot \min_{m\in M}\|\mu_j-m\|^2$ for all $p\in A_{i,M}$, as $\Gamma_i < 2\max_{p\in C_j}\min_{m\in M}\|p-m\|$.
Since for any cluster $C_{j'}$ with $j'>j$ we have $|C_{j'}|\leq |C_j|$ and therefore 
\begin{eqnarray}
  \sum_{p\in C_{j'}\cap A_{i,M}} \|p-m\|^2 
&\leq& |C_{j'}\cap A_{i,M}| \cdot 8\cdot \frac{\beta}{\varepsilon} \cdot 
   \min_{m\in M}\|\mu_j-m\|^2\leq |C_j| \cdot 8\cdot \frac{\beta}{\varepsilon} \cdot \min_{m\in M}\|\mu_j-m\|^2\nonumber \\
&\leq& 2176\cdot \left(\frac{\beta}{\varepsilon}\right)^{5/2} \cdot \sum_{p\in C_j^\beta}\min_{m\in M}\|p-m\|^2.
\label{eq: smaller clusters}
\end{eqnarray}

Define ${\cal H} = \cup_{h=1}^{j-1} C_h$ and $\mathcal{L} = \cup_{h=j+1}^k C_h$.
Clearly 
$$\mathbb{P}_i[p\in (\mathcal{H}\cup C_j\cup \mathcal{L})\cap A_{i,M}] = 1.$$
By Lemma~\ref{lem:lbprob}, 
$$
\mathbb{P}_i[p\in C_j^\beta\mid p\in (C_j\cup\mathcal{H})\cap A_{i,M}] \geq \frac{1}{1088\cdot k}\cdot
 \left(\frac{\varepsilon}{\beta^{3/8}}\right)^4.
$$
By Inequality~\eqref{eq: smaller clusters},
$$\mathbb{P}_i[p\in C_j^\beta\mid p\in (C_j\cup\mathcal{L})\cap A_{i,M}] \geq \frac{1}{2176 \cdot k} \cdot \left(\frac{\varepsilon}{\beta}\right)^{5/2}.$$
Now,
$$\max\left\{\mathbb{P}_i[p\in (\mathcal{H}\cup C_j)\cap A_{i,M}], 
\mathbb{P}_i[p\in (C_j\cup\mathcal{L})\cap A_{i,M}]\right\}\geq \frac 1 2,$$ 
so, 
\begin{align*}
 \mathbb{P}_i[p\in C_j^\beta] &= \mathbb{P}_i[p\in C_j^\beta \mid p\in (\mathcal{H}\cup C_j)\cap A_{i,M}]\cdot \mathbb{P}_i[p\in (\mathcal{H}\cup C_j)\cap A_{i,M}] \\  
 &= \mathbb{P}_i[p\in C_j^\beta \mid p\in (C_j\cup\mathcal{L})\cap A_{i,M}]\cdot 
  \mathbb{P}_i[p\in (C_j\cup\mathcal{L})\cap A_{i,M}] \\
 &\geq \frac 1 2 \cdot \min\left\{\mathbb{P}_i[p\in C_j^\beta | p\in (\mathcal{H}\cup C_j)\cap A_{i,M}], \mathbb{P}_i[p\in C_j^\beta | p\in (C_j\cup\mathcal{L})\cap A_{i,M}]\right\} \\
 &\geq \frac{1}{2}\min\left\{\frac{1}{2176 \cdot k} \cdot \left(\frac{\varepsilon}{\beta}\right)^{5/2},
 \frac{1}{1088\cdot k}\cdot
 \left(\frac{\varepsilon}{\beta^{3/8}}\right)^4\right\} \geq \frac{1}{4352\cdot k}\cdot \left(\frac{\varepsilon}{\beta^{5/8}}\right)^4.
\end{align*}
We remark that by Lemma~\ref{lem:preprocess} we may assume that all non-zero squared distances are within a factor $n^{30}$ of each other. Thus, the desired $i < 30\log n$.
\end{proof}

Finally, we show that we can account for the bias in the sampling in order to estimate an approximate mean.
\begin{lemma}
\label{lem:weightedmeans}
    Let $M$ be a set of centers. Let $j$ be the smallest index such that $C_j$ is not $\varepsilon$-covered by $M$. 
    Suppose that $M$ is not an $(\varepsilon/4,\beta)$-mean seeding set for $C_j$.
    Consider a set of points $S'$ sampled iid from $\mathbb{P}_i$, and let $S = S'\cap C_j^\beta$.
    If $\beta \geq 2400\varepsilon^{-2}$ and $S> 17825792 \cdot k\left(\frac{\beta^{7/12}}{\varepsilon}\right)^6 \log(2/\delta)$, then 
    with probability at least $1-\delta$, we have that $S'\cup M$ is an $(\varepsilon/4,\beta)$-mean seeding set of $C_j$.
\end{lemma}
\begin{proof}
    We first apply some preprocessing. Let $q$ be an arbitrary point in $S$. We subtract $q$ from all points. 
    Therefore, we may assume that all points $p\in C_j^\beta$, as well as any point $m\in M$ that has distance at most $\sqrt{\varepsilon^2 \Delta_j/2}$ to some point in $C_j^\beta$ have norm at most $\sqrt{(\beta+\varepsilon^2/2)\Delta_j}$.

    Furthermore, let $\mu_D$ be the mean of $D(C_j,\beta,\varepsilon/16,M)$, and let $\mu_C$ be the mean of $C=C_j^\beta\setminus D(C_j,\beta,\varepsilon/16,M)$.
    Due to Lemma \ref{lem:core}, we have that $\hat{\mu_j} = \frac{1}{|C_j^\beta|}\cdot\left(\mu_C\cdot |C| + \mu_D\cdot |D(C_j,\beta,\varepsilon/16,M)|\right)$ is an $\frac{\varepsilon}{4}$-approximate mean of $\mu_j$, or more specifically
    $$\|\hat{\mu_j}-\mu_j\|\leq \sqrt{\frac{\varepsilon}{4}\cdot \Delta_j}.$$
    Thus, if we can show that $S$ is an $\frac{\varepsilon}{4}$-mean seeding set of $\mu_C$ (yielding an
    $\frac{\varepsilon}{4}$-approximate mean $\widehat{\mu_C}$, then
    \begin{eqnarray*}
    & &\left\|\frac{1}{|C_j^\beta|}\left(\widehat{\mu_C}\cdot |C| + \mu_D\cdot |D(C_j,\beta,\varepsilon/16,M)|\right)-\mu_j\right\|  \\
    &\leq&  \|\widehat{\mu_C}-\mu_C\| + \left\|\frac{1}{|C_j^\beta|} \left(\mu_C\cdot |C| + \mu_D\cdot |D(C_j,\beta,\varepsilon/16,M)|\right)-\mu_j\right\|\\
    &\leq& \sqrt{\varepsilon/4\Delta_j} + \sqrt{\varepsilon/4\Delta_j} \leq \sqrt{\varepsilon\Delta_j},
    \end{eqnarray*}
    where we used Lemma~\ref{lem:magic} in the first inequality.

    Let $i$ to denote the largest index for which $A_{i,M}$ contains $C_j^\beta$.
    Define for every point $p\in C$ a weight $w_p = \frac{1}{|C|\cdot\mathbb{P}_i[p\mid C]}$. To clarify, $\mathbb{P}_i[p\mid C]$ 
    is the conditional probability that a single sample drawn from the probability distribution $\mathbb{P}_i$ is $p$, 
    conditioned on the sampled point being from $C$. We can then write
    $$
    \mu_C = \sum_{p\in C} (w_pp)\cdot\mathbb{P}_i[p\mid C].
    $$
    In other words, $\mu_C$ is the expectation of the scaled vector $w_pp$ under the conditional distribution $\mathbb{P}_i[\cdot\mid C]$
    Let $S_C = S\cap C$. Conditioning on $s = |S\cap C|$, the sample $S_C$ can be generated by taking $s$ independent samples from
    the distribution $\mathbb{P}_i[\cdot\mid C]$. We write $S_C = \{p_1,p_2,\dots,p_s\}$, where the points are random variables. Define
    $$
    \widehat{\mu_C} = \frac{1}{s}\cdot\sum_{p\in S_C} w_pp.
    $$
    Taking expectation over $\mathbb{P}_i[\cdot\mid C\wedge s]$, we have
    \begin{eqnarray*}
        {\mathbb E}\left[\|\widehat{\mu_C} - \mu_C\|^2\right] & = &
        {\mathbb E}\left[\frac{1}{s^2}\cdot\sum_{i=1}^s\sum_{j=1}^s (w_{p_i}p_i - \mu_C)\cdot (w_{q_j}q_j - \mu_C)\right] \\
        & = & \frac{1}{s^2}\cdot\sum_{i=1}^s {\mathbb E}\left[\|w_{p_i}p_i - \mu_C\|^2\right].
    \end{eqnarray*}
The cross terms vanish as the sampled points are independent and the expectation of $w_p p$ is $\mu_C$.

To complete the proof, notice that for $p\in C$, $\|w_{p}p - \mu_C\|^2\le 2w_p\|p\|^2 + 2\|\mu_C\|^2$. We may
assume without loss of generality that the entire point-set is shifted so that $\mu_C = \vec{0}$. Hence, as
$\mu_C\in \hbox{conv}(C_j^\beta)$ and $p\in C_j^\beta$, we have that $\|p\|^2\le 4\beta\Delta_j$. Also, 
$\frac{\varepsilon}{16}\Delta_j\leq \min_{m\in M} \|p-m\|^2 \leq \beta \cdot \Delta_j$, 
where the lower bound holds by definition of $D(C_j,\beta,\varepsilon/16,M)$ and the upper bound holds by 
definition of $C_j^\beta$. Thus, $\frac{\varepsilon}{16\cdot |C|}\leq \mathbb{P}_i[p\mid p\in C] \leq \frac{\beta}{|C|}$.
This implies that $w_p\le \frac{16}{\varepsilon}$. Therefore,
$$
{\mathbb E}\left[\|\widehat{\mu_C} - \mu_C\|^2\right] \le \frac{1}{s}\cdot\frac{64\beta}{\varepsilon}\cdot\Delta_j,
\hbox{ so  } \mathbb{P}_i\left[\|\widehat{\mu_C} - \mu_C\|^2 > \frac{1}{s}\cdot\frac{128\beta}{\varepsilon}\cdot\Delta_j\right] < \frac 1 2.
$$
If $s\ge \frac{512\beta}{\varepsilon^2}$, we get that $\widehat{\mu_C}$ $\frac{\varepsilon}{4}$-covers $\mu_C$ with probability
at least $\frac 1 2$. Thus, if $s\ge \frac{512\beta}{\varepsilon^2}\cdot\log(2/\delta)$ we can apply this $\log\delta^{-1}$
times to boost the success probability to $1-\frac{\delta}{2}$.

We now bound the number of samples that we need to obtain $S_C$.
Due to Lemma \ref{lem:refinedsample}, we have $\mathbb{P}_i[p\in C_j^\beta]\geq \frac{1}{4352\cdot k}\cdot \left(\frac{\varepsilon}{\beta^{5/8}}\right)^4$. 
Therefore, $\mathbb{E}_i[|S_C|] = |S| \cdot \mathbb{P}_i[p\in C_j^\beta] \geq |S|\cdot \frac{1}{4352\cdot k}\cdot \left(\frac{\varepsilon}{\beta^{5/8}}\right)^4$.
Setting $|S|\geq 17825792 \cdot k\left(\frac{\beta^{7/12}}{\varepsilon}\right)^6 \log(2/\delta)$ and applying the Chernoff bound, we have
$$\mathbb{P}\left[|S_C|< \frac{512\beta}{\varepsilon^2}\cdot\log(2/\delta)\right] \leq \exp(-8 \cdot \mathbb{E}[|S_C|]]) \leq \delta/2.$$
Conversely, with probability $1-\delta$, $S\cup M$ contains a $(\varepsilon/4,\beta)$-mean seeding set of $C_j$.
\end{proof}

We are now ready to give a proof of Lemma \ref{lem:approxmeans}.

\begin{proof}[Proof of Lemma \ref{lem:approxmeans}]
Due to Lemma \ref{lem:preprocess}, we know that we have at most $k$ point $A_1,\ldots A_k$ sets such that any cluster of the optimum clustering is fully contained in one the $A_i$.
We guess the correct number of centers from each $A_i$, which takes at most ${2k-1 \choose k-1}$ guesses.
For each $A_i$, we then find a set of centers $M$ that $\varepsilon$ covers all clusters of the optimum in $A_i$.

We simplify the calculation by assuming that $A_i$ contains all $k$ clusters.
We iteratively add centers to $M$, writing $M_j$ after the $j$-th iteration. Our goal is to ensure that $M_j$ covers the clusters $C_1,\ldots C_j$.
In every iteration, we first sample to obtain a suitable mean seeding set and then apply Lemma \ref{lem:constructive} to extract the mean from the set.

We start with $C_1$. We know that $|C_1|\geq \frac{n}{k}$, so we can use Lemma \ref{lem:uniform} to sample a set $S_1$ of $32 k\varepsilon^{-1}\log(k/\delta)$ points uniformly at random and then enumerate over all candidate means induced by uniformly weighted subsets of $S_1$ and the to obtain an $\varepsilon$-approximate mean of $C_1$. This takes time $2^{|S_1|}$  and yields $2^{|S_1|}$ candidate means, of which one is an $\varepsilon$-covers $C_1$ with probability $1-\delta/k$.

For subsequent iterations, Lemma \ref{lem:weightedmeans} guarantees us that there exists a distribution $\mathbb{P}_i$ such that if we sample a set $S_j$ of $17825792 \cdot k\left(\frac{\beta^{7/12}}{\varepsilon}\right)^6 \log(2k/\delta)$ points, then $M_{j-1}\cup S$ is an $(\varepsilon/4,\beta)$ mean seeding set of $C_j$ with probability $1-\delta/k$. Moreover, Lemma \ref{lem:preprocess2} guarantees us that we have to try at most $13\log(nk/\varepsilon)$ distributions to do find the correct $\mathbb{P}_i$.
Extracting all candiate means for each $\mathbb{P}_i$ via Lemma \ref{lem:constructive} takes time $\left(\frac{10\beta\cdot |S_j|}{\varepsilon}+1\right)^{|S_j|}$ and results in $\left(\frac{10\beta\cdot |S_j|}{\varepsilon}+1\right)^{|S_j|}$ candidate means.

Thus, the overall number of candidate centers $M_k$ generated by the procedure, as well as the running time, is
$$2^{|S_1|}\cdot \prod_{j=2}^k 13\log(nk/\varepsilon) \cdot \left(\frac{10\beta\cdot |S_j|}{\varepsilon}+1\right)^{|S_j|} = \log^k n \cdot 2^{\eta\cdot k^2 \cdot \varepsilon^{-12} \log^2(k/(\varepsilon\delta))}$$
for some absolute constant $\eta$.
Moreover by the union bound, one of the $M_k$ must $\varepsilon$ cover all clusters with probability $1-\delta$.
Notice that if $\log n < k^2$, then $\log^k n$ is absorbed by $2^{\eta\cdot k^2 \cdot \varepsilon^{-12} \log^2(k/(\varepsilon\delta))}$ with a suitable rescaling of $\eta$. If $\log n>k^2$, then $\log^k n <2^{\sqrt{\log n}\log\log n}<n^{o(1)}$.

We account for the enumeration over the number of clusters from each $A_i$ via another rescaling of $\eta$.
For a given $M$ and $\mathbb{P}_i$, the probabilities can be computed in time $O(n\cdot d\cdot |M|)$
Thus, the overall running time to obtain a set of centers that $\varepsilon$ covers all clusters of the optimum is
$$n^{1+o(1)} \cdot d \cdot 2^{\eta\cdot k^2 \cdot \varepsilon^{-12} \log^2(k/(\varepsilon\delta))},$$
and this completes the proof.
\end{proof}

\paragraph{Enumerating over Sizes and Obtaining the Parameterized PTAS.}
We complete this section by funneling the mean-seeding procedure into a PTAS.

\begin{theorem}
    There exists an algorithm running in time
    $$O\left(n^{1+o(1)}d\cdot  2^{\eta\cdot k^2 \cdot \varepsilon^{-12} \log^2(k/(\varepsilon\delta))}\right),$$ 
    for some absolute constant $\eta$, that computes a $(1+\varepsilon)$-approximate solution to $\ell_2^2$ $k$-MinSum Clustering with probability $1-\delta$.
\end{theorem}
\begin{proof}
Given a set of candidate centers obtained via Lemma \ref{lem:approxmeans} and an estimate $\widehat{\OPT}$ of the optimal MinSum clustering cost $\OPT$, we wish to find an assignment of points to centers such that the clustering that has cost $(1+\varepsilon)\cdot\widehat{\OPT}$, or verify that no such assignment exists.
Note that given a clustering, we can verify its cost in time $O(ndk)$ by computing the mean of every cluster and then using the first identity of Lemma \ref{lem:magic}.

We first notice that if we are given an $\alpha$-approximation $\widehat{\OPT_{kmeans}}$ to an $k$-means clustering $\OPT_{kmeans}$, we also know $\OPT\in \left[\widehat{\OPT_{kmeans}},n\cdot \widehat{\OPT_{kmeans}}\right]$. A constant, say $20$, approximation to $k$-means can be found in time $\tilde{O}(nd + k^2)$ \cite{DraganovSS24}. We thus can efficiently obtain  $(1+\varepsilon)$ approximate value of $\OPT$ using at most $2\varepsilon^{-1}\log(20n)$ estimates.

Suppose we are given $\widehat{\OPT}$, as well as a candidate set of centers $C=\{c_1,c_2,\ldots c_k\}$.
Now, we discretize the cost of all points to each cluster $c_i$, starting at $\frac{\varepsilon}{n^2}\cdot \widehat{\OPT}$ by powers of $(1+\varepsilon)$, going all the way up to $\widehat{\OPT}$. Define 
$$G_{i,j}=\left\{p~|~ (1+\varepsilon)^{j-1} \cdot \frac{\varepsilon}{n^2}\cdot \widehat{\OPT}\leq \|p-c_i\|^2 \leq (1+\varepsilon)^{j} \cdot \frac{\varepsilon}{n^2}\cdot\widehat{\OPT}\right\}$$ 
with $G_{i,0} = \left\{p~|~ \|p-c_i\|^2 \leq \frac{\varepsilon}{ n}\cdot\widehat{\OPT}\right\}$. Notice that if $\|p-c_i\|^2>(1+\varepsilon)\cdot \widehat{\OPT}$, then $p$ cannot be served by $c_i$ without invalidating $\widehat{OPT}$ as an accurate estimate of $\OPT$. Thus we have at most $2\varepsilon^{-1}\log \frac{n^2}{\varepsilon}$ many sets $G_{i,j}$. Finally, consider the set 
$B_{1_{j'},2_{j''},1_{j'''},\ldots }$ which is the intersection of $G_{1,j}\cap G_{2,j'}\cap G_{2,j''}\ldots$.
Notice that there are $\left(2\varepsilon^{-1}\log \frac{n^2}{\varepsilon}\right)^k$ many sets $B$ and that we can compute the partitioning of the point set $A$ into the sets $B$ in time $ndk\cdot \left(2\varepsilon^{-1}\log \frac{n^2}{\varepsilon}\right)^k$.
We finally discretize the size of subsets of any set $B$ by powers of $(1+\varepsilon)$, for which there are $2\varepsilon^{-1}\log |B| \leq 2k\varepsilon^{-1}\log n$ discretizations. 

We now enumerate over all possible assignments of subsets of sets $B$ to centers $c_i$. Notice that there are at most $\left(2\varepsilon^{-1}\log n\right)^k$ possible sizes, which we multiply by the number $\left(2\varepsilon^{-1}\log \frac{n}{\varepsilon}\right)^k$ of sets $B$.

We claim that if $C$ is the center set of a $(1+\varepsilon)$-approximate solution, then there exists an assignment of the $B$ that is $(1+O(\varepsilon))$ approximate as well. Specifically, consider any assignment $\pi:A\rightarrow C$ cost $\text{cost}_\pi(A,C) = \sum_{p\in A}\|p-\pi(p)\|^2$.
In the following, we use $B_j$ to refer to the intersection of $B$ with $C_i$, i.e. $B_{i,j} =C_i \cap B_{1_{j'},2_{j''},\ldots, }$
Then rewriting the sum, we obtain
\begin{eqnarray*}
\sum_{p \in C_i}(\sum_{j} |B_{i,j}|)\sum_{j>0}|B_{i,j}|\cdot 2^{j-1} \frac{\varepsilon}{n^2}\cdot \widehat{\OPT}
   &\leq &\text{cost}_\pi(A,C) \\
   & \leq & (1+\varepsilon)\cdot \sum_{p \in C_i}(\sum_{j} |B_{i,j}|)\sum_{j>0}|B_{i,j}|\cdot 2^{j} \frac{\varepsilon}{n^2}\cdot \widehat{\OPT} + n^2 \cdot \frac{\varepsilon}{n^2}\cdot \widehat{\OPT}\\
   & =& (1+\varepsilon)\cdot \sum_{p \in C_i}(\sum_{j} |B_{i,j}|)\sum_{j>0}|B_{i,j}|\cdot 2^{j} \frac{\varepsilon}{n^2}\cdot \widehat{\OPT} + \varepsilon\cdot \widehat{\OPT}
\end{eqnarray*}
and moreover
$$(1+\varepsilon)\cdot \sum_{p \in C_i}\left(\sum_{j} |C_{i,j}|\right)\sum_{j>0}|C_{i,j}|\cdot 2^{j} \frac{\varepsilon}{n^2}\cdot \widehat{\OPT}  \leq (1+\varepsilon)\cdot \sum_{p \in C_i}\left(\sum_{j} |C_{i,j}|\right)\sum_{j>0}|C_{i,j}|\cdot 2^{j-1} \frac{\varepsilon}{n^2}\cdot \widehat{\OPT}.$$
In other words, using the discretizations $B$ instead of the correct points in the assignment of $A$ to $C$ preserves the cost up to a multiplicative factor $(1+\varepsilon)$ and an additive $\varepsilon\cdot \widehat{\OPT}$.

Next, observe that if we have an estimate $|B_{i,j}|\leq 
 \hat{B_{i,j}} \leq (1+\varepsilon)\cdot |B_{i,j}|$, then
$\sum_{j}|B_{i,j}|  \leq \sum_{j}\hat{B_{i,j}} \leq (1+\varepsilon)\cdot \sum_{j}|B_{i,j}|$. Therefore, using the discretized estimates of $|B_{i,j}|$, we also have

\begin{align*}
\sum_{p \in C_i}(\sum_{j} \hat{B_{i,j}})\sum_{j>0}\hat{B_{i,j}}|\cdot 2^{j-1} \frac{\varepsilon}{n^2}\cdot \widehat{\OPT} 
   &\leq \text{cost}_\pi(A,C) \\
    &\leq  (1+\varepsilon)^3 \sum_{p \in C_i}\left(\sum_{j} \hat{B_{i,j}}\right)\sum_{j>0}\hat{B_{i,j}}\cdot 2^{j} \frac{\varepsilon}{n^2}\cdot \widehat{\OPT} + \varepsilon\cdot \widehat{\OPT}
\end{align*}

Given a (discretized) assignment of the sets $B$ to $C$, we can now extract a clustering as follows.
In the following the value of $j$ is not necessary so we omit the subscript $j$ from $B_{i,j}$.
We sort $\hat{B_{i}}$ by sizes, breaking ties arbitrarily.
We assign $\hat{B_{i}}$ many arbitrary points of $B$ to cluster $C_i$ with center $c_i$. The final cluster $C_{i'}$ in the ordering is assigned the remaining points. Notice that assigning fewer points to $C_{i'}$ can only decrease the cost of $C_{i'}$. 

The cost of this assignment can only be cheaper than the estimated upper bound
$$(1+\varepsilon)^3 \sum_{p \in C_i}(\sum_{j} \hat{B_{i,j}})\sum_{j>0}\hat{B_{i,j}}\cdot 2^{j} \frac{\varepsilon}{n^2}\cdot \widehat{\OPT} + \varepsilon\cdot \widehat{\OPT},$$ 
as we can only assign fewer points from every group $B$ to a cluster and the cost of the points can only be cheaper than the estimated upper bound. As mentioned above, evaluating the cost of the resulting clustering takes time $O(ndk)$.

Thus, assuming that $\OPT\leq \hat{OPT}\leq (1+\varepsilon)\cdot \OPT$ and that we were working with a suitable $\varepsilon$-approximate candidate set of centers $C$, we can extract a clustering with cost at most $(1+\varepsilon)^5 \cdot \OPT$ in time $O\left(nd \left(2\varepsilon^{-1}\log n\right)^k\cdot \left(2\varepsilon^{-1}\log \frac{n}{\varepsilon}\right)^k\right)$
multiplying this figure by the number of candidate values of $OPT$ and the number of candidate centers obtained via Lemma \ref{lem:approxmeans} yields a running time of
$$O\left(nd  \cdot \left(2\varepsilon^{-1}\log \frac{20n}{\varepsilon}\right)^{3k}\cdot n^{o(1)} \cdot 2^{\eta\cdot k^2 \cdot \varepsilon^{-12} \log^2(k/(\varepsilon\delta))}\right)$$ 
plus the running time for computing the candidate centers. Using $(\log n)^k \leq k^{3k} + 2^{\sqrt{\log n}\log\log n}\leq k^{3k}+n^{o(1)}$, rescaling $\varepsilon$ by a factor of $10$, this yields a $(1+\varepsilon)$ approximation with probability $1-\delta$ in time
$$O\left(n^{1+o(1)}d\cdot 2^{\eta\cdot k^2 \cdot \varepsilon^{-12} \log^2(k/(\varepsilon\delta))}\right)$$ 
for some absolute constant $\eta$.
\end{proof}

\section{Learning-Augmented \texorpdfstring{$\ell_2^2$}{L22} Min-Sum \texorpdfstring{$k$}{k}-Clustering}
\seclab{sec:learn:aug}
In this section, we describe and analyze our learning-augmented algorithm for $\ell_2^2$ min-sum $k$-clustering, corresponding to \thmref{thm:learn:aug}. 


We first recall the following property describing the $1$-means optimizer for a set of points. 
\begin{fact}
\factlab{fact:means:center}
\cite{InabaKI94}
Given a set $X\subset\mathbb{R}^d$ of points, the unique minimizer of the $1$-means objective is 
\[\frac{1}{|X|}\sum_{x\in X}x=\argmin_{c\in\mathbb{R}^d}\sum_{x\in X}\|x-c\|_2^2.\]
\end{fact}
We next recall the following identity, which presents an equivalent formulation of the $\ell_2^2$ min-sum $k$-clustering objective. 
\begin{fact}
\cite{InabaKI94}
\factlab{fact:minsum}
For each cluster $C_i$ of points, let $c_i$ be the geometric mean of the points, i.e., 
\[c_i=\frac{1}{|C_i|}\sum_{x\in C_i} x.\]
Then
\[\frac{1}{2}\sum_{i\in[k]}\sum_{x_u,x_v\in C_i}\|x_u-x_v\|_2^2=\sum_{i\in[k]}|C_i|\cdot\sum_{x\in C_i}\|x-c_i\|_2^2.\]
\end{fact}
Given \factref{fact:minsum}, it is more convenient for us to rescale the $\ell_2^2$ min-sum $k$-clustering objective for an input set $X$ in this section to be defined as:
\[\min_{C_1,\ldots,C_k}\frac{1}{2}\sum_{i\in[k]}\sum_{p,q\in C_i\cap X}\|p-q\|_2^2.\]
We now formally define the precision and recall guarantees of a label predictor. 
\begin{definition}[Label predictor]
\deflab{def:label:oracle}
Suppose that there is an oracle that produces a label $i\in[k]$ for each $x\in X$, so that the labeling partitions $X=P_1\dot\cup\ldots\dot\cup P_k$ into $k$ clusters $P_1,\ldots,P_k$, where all points in $P_i$ have the same label $i\in[k]$. 
We say the oracle is a \emph{label predictor with error rate} $\alpha$ if there exists some fixed optimal min-sum clustering $P^*_1,\ldots,P^*_k$ such that for all $i\in[k]$,
\[|P_i\cap P^*_i|\ge(1-\alpha)\max(|P_i|,|P^*_i|).\]
We say that $P^*=\{P^*_1,\ldots,P^*_k\}$ is the clustering consistent with the label oracle. 
\end{definition}
We also recall the following guarantees of previous work on learning-augmented $k$-means clustering for a label predictor with error rate $\alpha\in\left[0,\frac{1}{2}\right)$. 
\begin{theorem}
\cite{NguyenCN23}
\thmlab{thm:each:cluster}
Given a label predictor with error rate $\alpha<\frac{1}{2}$ consistent with some clustering $P^*=\{P^*_1,\ldots,P^*_k\}$ with centers $\{c^*_1,\ldots,c^*_k\}$, there exists a polynomial-time algorithm $\LearnedCenters$ that outputs a set of centers $\{c_1,\ldots,c_k\}$, so that for each $i\in[k]$,
\[\sum_{x\in P^*_i}\|x-c_i\|_2^2\le(1+\gamma_\alpha\alpha)\sum_{x\in P^*_i}\|x-c^*_i\|_2^2,\]
where $\gamma_\alpha=7.7$ for $\alpha\in\left[0,\frac{1}{7}\right)$ or $\gamma_\alpha=\frac{5\alpha-2\alpha^2}{(1-2\alpha)(1-\alpha)}$ for $\alpha\in\left[0,\frac{1}{2}\right)$. 
\end{theorem}

\paragraph{Description of $\LearnedCenters$.}
For the sake of completeness, we briefly describe the algorithm $\LearnedCenters$ underlying \thmref{thm:each:cluster}. 
The algorithm decomposes the $k$-means clustering objective by considering the subset $P_i$ of the input dataset $X$ that are assigned each label $i\in[k]$ by the oracle. 
The algorithm further decomposes the $k$-means clustering objective along the $d$ dimensions, by considering the $j$-th coordinate of each subset $P_i$, for each $j\in[d]$. 
Now, although an $\alpha$ fraction of the points in $P_i$ can be incorrectly labeled, there are two main cases: 1) $P_i$ includes a number of mislabeled points that are far from the true mean and hence easy to prune away, or 2) $P_i$ includes a number of mislabeled points that are difficult to identify due to their proximity to the true mean. 
However, in the latter case, these mislabeled points only has a small effect on the overall $k$-means clustering objective. 
Hence, it suffices for the algorithm to handle the first case, which it does by selecting the interval of $(1-\O{\alpha})$ points of $P_i$ in dimension $j$ that has the best clustering cost. 
The mean of the points of $P_i$ in dimension $j$ that lie in that interval then forms the $j$-th coordinate of the $i$-th centroid output by algorithm. 
The algorithm repeats across $j\in[d]$ and $i\in[k]$ to form $k$ centers that are well-defined in all $d$ dimensions. 
We give the algorithm formally in \algref{alg:aug:kmeans}. 

\begin{algorithm}[!htb]
\caption{$\LearnedCenters$: learning-augmented $k$-means clustering \cite{NguyenCN23}}
\alglab{alg:aug:kmeans}
\begin{algorithmic}[1]
\Require{Dataset $X$ with partition $P_1,\ldots,P_k$ induced by label predictor with error rate $\alpha$}
\Ensure{Centers $c_1,\ldots,c_k$ for $(1+\O{\alpha})$-optimal $k$-means clustering}
\For{$i\in[k]$}
\For{$j\in[d]$}
\State{Let $\omega_{i,j}$ be the collection of all intervals that contain $(1-\O{\alpha})|P_i|$ points of $P_{i,j}$}
\State{Let $c_{i,j}$ be the center with the lowest $k$-means clustering cost of any interval in $\omega_{i,j}$ }
\EndFor
\State{$c_i\gets\{c_{i,j}\}_{j\in[d]}$ for all $i\in[d]$}
\State{\Return $\{c_1,\ldots,c_k\}$}
\EndFor
\end{algorithmic}
\end{algorithm}

By \factref{fact:means:center} and \factref{fact:minsum}, it follows that these centers are also good centers for the clustering induced by a near-optimal $\ell_2^2$ min-cost $k$-clustering.  
Specifically, the optimal center of a cluster of points for $\ell_2^2$ min-cost $k$-clustering is the centroid of the cluster and similarly, the optimal center of a cluster of points for $k$-means clustering is the centroid of the cluster. 
See \lemref{lem:flow:to:cost:two} for the formal details. 

Unfortunately, although the centers $\{c_1,\ldots,c_k\}$ returned by $\LearnedCenters$ are good centers for the clustering induced by a near-optimal $\ell_2^2$ min-cost $k$-clustering, it is not clear what the resulting assignment should be. 
In fact, we emphasize that unlike $k$-means clustering, the optimal $\ell_2^2$ min-cost $k$-clustering may not assign each point to its closest center. 

\paragraph{Constrained min-cost flow.}
To that end, we now create a constrained min-cost flow problem as follows. 
We first create a source node $s$ and a sink node $t$ and require that $n=|X|$ flow must be pushed from $s$ to $t$. 
We create a node $u_x$ for each point $x\in X$ and create a directed edge from $s$ to each node $u_x$ with capacity $1$ and cost $0$. 
There are no more outgoing edges from $s$ or incoming edges to each $u_x$. 
This ensures that to achieve $n$ flow from $s$ to $t$, a unit of flow must be pushed across each node $u_x$. 

For each center $c_i$ output by our learning-augmented algorithm, we create a node $v_i$. 
For each $x\in X$, $i\in[k]$, create a directed edge from $u_x$ to $v_i$ with capacity $1$ and cost $\frac{1}{1-\alpha}\cdot|P_i|\cdot\|x-c_i\|_2^2$. 
There are no other outgoing edges from $u_x$, thus ensuring that a unit of flow must exit each node $u_x$ to the nodes $v_i$ representing the clusters, and with approximately the corresponding cost if $x$ were assigned to center $c_i$. 
We then create a directed edge from each node $v_i$ to $t$ with capacity $\frac{1}{1-\alpha}\cdot|P_i|$ and cost $0$. 
Finally, we require that at least $(1-\alpha)\cdot|P_i|$ flow goes through node $v_i$, so that the number of points assigned to each center $c_i$ is consistent with the oracle.  
The construction in its entirety appears in \figref{fig:flow}. 

\begin{figure*}[!htb]
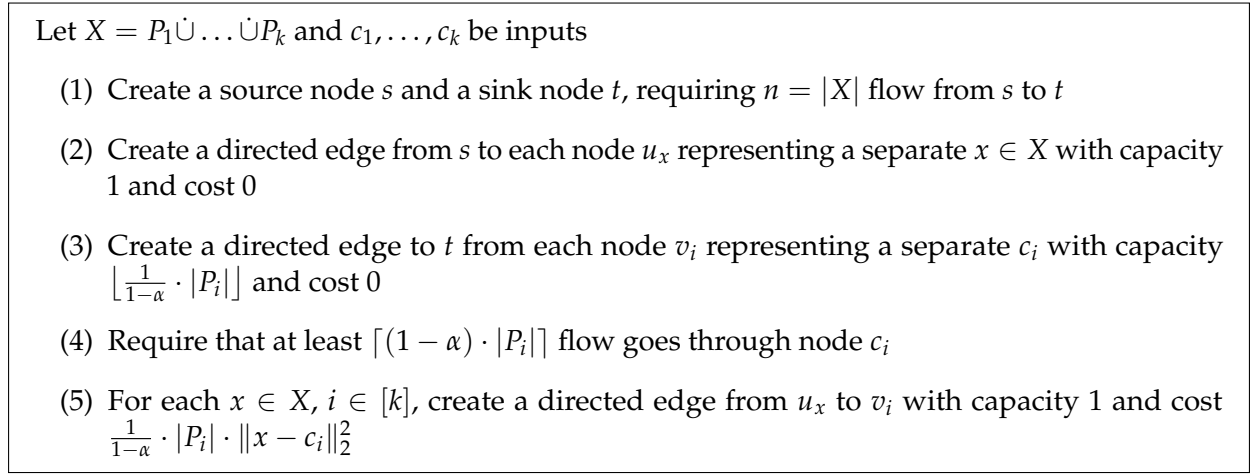

\begin{mdframed}
Let $X=P_1\dot\cup\ldots\dot\cup P_k$ and $c_1,\ldots,c_k$ be inputs
\begin{enumerate}
\item 
Create a source node $s$ and a sink node $t$, requiring $n=|X|$ flow from $s$ to $t$
\item 
Create a directed edge from $s$ to each node $u_x$ representing a separate $x\in X$ with capacity $1$ and cost $0$
\item 
Create a directed edge to $t$ from each node $v_i$ representing a separate $c_i$ with capacity $\flr{\frac{1}{1-\alpha}\cdot|P_i|}$ and cost $0$
\item 
Require that at least $\ceil{(1-\alpha)\cdot|P_i|}$ flow goes through node $c_i$
\item
For each $x\in X$, $i\in[k]$, create a directed edge from $u_x$ to $v_i$ with capacity $1$ and cost $\frac{1}{1-\alpha}\cdot|P_i|\cdot\|x-c_i\|_2^2$
\end{enumerate}
\end{mdframed}
\caption{Constrained min-cost flow problem}
\figlab{fig:flow}
\end{figure*}

\begin{algorithm}[!htb]
\caption{Learning-augmented min-sum $k$-clustering}
\alglab{alg:aug:minsum}
\begin{algorithmic}[1]
\Require{Dataset $X$ with partition $P_1,\ldots,P_k$ induced by label predictor with error rate $\alpha$}
\Ensure{Labels for all points consistent with a $(1+\O{\alpha})$-optimal min-sum $k$-clustering}
\State{Let $c_1,\ldots,c_k$ be the output centers of $\LearnedCenters$ on $P_1,\ldots,P_k$}
\State{Create a min-cost flow problem $\calF$ with required flow $n$ as in \figref{fig:flow}}
\State{Solve the flow problem $\calF$}
\State{For each $x\in X$, let the flow from $u_x$ be sent to the node $v_{\ell_x}$, so that $\ell_x\in[k]$}
\State{Label $x$ with $\ell_x$}
\end{algorithmic}
\end{algorithm}

We first show that the $\ell_2^2$ min-sum $k$-clustering cost induced by \algref{alg:aug:minsum} has objective value at most the cost of the optimal flow in the problem $\calF$ created by \algref{alg:aug:minsum}.  
\begin{lemma}
\lemlab{lem:flow:to:cost:one}
Let $F$ be the cost of the flow output by \algref{alg:aug:minsum}. 
Then for the corresponding clustering $Q_1,\ldots,Q_k$ output by \algref{alg:aug:minsum}, we have
\[\frac{1}{2}\sum_{i\in[k]}\sum_{x_u,x_v\in Q_i}\|x_u-x_v\|_2^2\le F.\]
\end{lemma}
\begin{proof}
Let $\calS$ be any flow output by \algref{alg:aug:minsum} and let $Q_1,\ldots,Q_k$ be the corresponding clustering of $X$. 
Note that $Q_1,\ldots,Q_k$ are well-defined, since each point of $x$ receives exactly one label by \algref{alg:aug:minsum}. 
Let $q_1,\ldots,q_k$ be the geometric mean of the points in $Q_1,\ldots,Q_k$, respectively, so that $q_i=\frac{1}{|Q_i|}\sum_{x \in Q_i}x$ for all $i\in[k]$. 

By \factref{fact:means:center} and \factref{fact:minsum}, we have that
\begin{align*}
\frac{1}{2}\sum_{i\in[k]}\sum_{x_u,x_v\in Q_i}\|x_u-x_v\|_2^2&=\sum_{i\in[k]}|Q_i|\cdot\sum_{x\in Q_i}\|x-q_i\|_2^2\\
&\le\sum_{i\in[k]}|Q_i|\cdot\sum_{x\in Q_i}\|x-c_i\|_2^2.
\end{align*}
Since each node $v_i$ has capacity $\frac{1}{1-\alpha}\cdot|P_i|$, then we have $|Q_i|\le\frac{1}{1-\alpha}\cdot|P_i|$. 
Therefore,
\[\frac{1}{2}\sum_{i\in[k]}\sum_{x_u,x_v\in Q_i}\|x_u-x_v\|_2^2\le\sum_{i\in[k]}\frac{1}{1-\alpha}\cdot|P_i|\cdot\sum_{x\in Q_i}\|x-c_i\|_2^2.\]
Because each $x\in Q_i$ is mapped to $c_i$, then the cost induced by the mapping in the flow $\calS$ is exactly $\frac{1}{1-\alpha}\cdot|P_i|\cdot\|x-c_i\|_2^2$. 
Therefore, the right-hand side is exactly the cost $F$ of the flow $\calS$. 
Hence, we have
\[\frac{1}{2}\sum_{i\in[k]}\sum_{x_u,x_v\in Q_i}\|x_u-x_v\|_2^2\le F,\]
as desired.
\end{proof}

We next show that the cost of the optimal $\ell_2^2$ min-sum $k$-clustering has objective value at least the cost of the optimal in the problem $\calF$ created by \algref{alg:aug:minsum}, up to a $(1+\O{\alpha})$ factor.   
\begin{lemma}
\lemlab{lem:flow:to:cost:two}
Let $F$ be the cost of the optimal solution to the min-cost flow problem $\calF$ in \algref{alg:aug:minsum} and let $\OPT$ be cost of the optimal min-sum $k$-clustering on $X$. 
Let $\gamma_\alpha$ be the fixed constant from \thmref{thm:each:cluster}. 
Then  
\[\OPT\ge(1-\alpha)^2\cdot\frac{1}{1+\gamma_\alpha\alpha}\cdot F.\]
\end{lemma}
\begin{proof}
Let $P^*_1,\ldots,P^*_k$ be an optimal clustering consistent with the label oracle. 
Let $c^*_1,\ldots,c^*_k$ be the optimal centers for $P^*_1,\ldots,P^*_k$ respectively and let $c_1,\ldots,c_k$ be the $k$ centers output by \algref{alg:aug:minsum}. 

By the definition of the label oracle, we have 
\[|P_i\cap P^*_i|\ge(1-\alpha)\max(|P_i|,|P^*_i|),\]
so that
\[|P^*_i|\ge|P_i\cap P^*_i|\ge(1-\alpha)\max(|P_i|,|P^*_i|)\ge(1-\alpha)\cdot|P_i|.\]
Thus, by \factref{fact:minsum}, 
\begin{align*}
\frac{1}{2}\sum_{i\in[k]}\sum_{x_u,x_v\in P^*_i}\|x_u-x_v\|_2^2&=\sum_{i\in[k]}|P^*_i|\cdot\sum_{x\in P^*_i}\|x-c^*_i\|_2^2\\
&\ge\sum_{i\in[k]}(1-\alpha)\cdot|P_i|\cdot\sum_{x\in P^*_i}\|x-c^*_i\|_2^2\\
&=(1-\alpha)^2\sum_{i\in[k]}\frac{1}{1-\alpha}\cdot|P_i|\cdot\sum_{x\in P^*_i}\|x-c^*_i\|_2^2.
\end{align*}
Let $\gamma_\alpha$ be the fixed constant from \thmref{thm:each:cluster}. 
Then by \thmref{thm:each:cluster}, we have that
\[\sum_{x\in P^*_i}\|x-c^*_i\|_2^2\ge\frac{1}{1+\gamma_\alpha\alpha}\cdot\sum_{x\in P^*_i}\|x-c_i\|_2^2.\]
Therefore,
\[\frac{1}{2}\sum_{i\in[k]}\sum_{x_u,x_v\in P^*_i}\|x_u-x_v\|_2^2\ge(1-\alpha)^2\cdot\frac{1}{1+\gamma_\alpha\alpha}\cdot\sum_{i\in[k]}\frac{1}{1-\alpha}\cdot|P_i|\cdot\sum_{x\in P^*_i}\|x-c_i\|_2^2.\]
Note that since $|P_i|\ge|P_i\cap P^*_i|\ge(1-\alpha)\max(|P_i|,|P^*_i|)\ge(1-\alpha)\cdot|P^*_i|$, then we have $|P^*_i|\le\frac{1}{1-\alpha}\cdot|P^*_i|$. 
Thus a valid flow for $\calF$ would be to send $|P^*_i|$ units of flow across each $x\in P^*_i$. 
In other words, $\sum_{i\in[k]}\frac{1}{1-\alpha}\cdot|P_i|\cdot\sum_{x\in P^*_i}\|x-c_i\|_2^2$ is the cost of a valid flow for $\calF$. 

Therefore, by the optimality of the optimal min-cost flow, we have 
\[\sum_{i\in[k]}\frac{1}{1-\alpha}\cdot|P_i|\cdot\sum_{x\in P^*_i}\|x-c_i\|_2^2\ge F,\]
and so
\[\frac{1}{2}\sum_{i\in[k]}\sum_{x_u,x_v\in P^*_i}\|x_u-x_v\|_2^2\ge(1-\alpha)^2\cdot\frac{1}{1+\gamma_\alpha\alpha}\cdot F,\]
as desired.
\end{proof}
Putting together \lemref{lem:flow:to:cost:one} and \lemref{lem:flow:to:cost:two}, it follows that the cost of the clustering induced by \algref{alg:aug:minsum} is a good approximation to the optimal $\ell_2^2$ min-sum $k$-clustering. 
\begin{corollary}
\corlab{cor:aug:correctness}
Let $\gamma_\alpha$ be the fixed constant from \thmref{thm:each:cluster}. 
\algref{alg:aug:minsum} outputs a clustering $Q_1,\ldots,Q_k$ of $X$ such that
\[\frac{1}{2}\sum_{i\in[k]}\sum_{x_u,x_v\in Q_i}\|x_u-x_v\|_2^2\le\frac{1+\gamma_\alpha\alpha}{(1-\alpha)^2}\cdot\OPT,\]
where $\OPT$ is cost of an optimal min-sum $k$-clustering on $X$. 
\end{corollary}
\begin{proof}
Let $\calS$ be the flow output by \algref{alg:aug:minsum} and let $Q_1,\ldots,Q_k$ be the corresponding clustering of $X$. 
We again remark that $Q_1,\ldots,Q_k$ is a valid clustering of $X$, since each point of $x$ receives exactly one label by \algref{alg:aug:minsum}. 
The claim then follows from \lemref{lem:flow:to:cost:one} and \lemref{lem:flow:to:cost:two}. 
\end{proof}
We recall the following folklore integrality theorem for uncapacitated min-cost flow. 
\begin{theorem}
\thmlab{thm:integral:sol}
Any minimum cost network flow problem with integral demands has an optimal solution with integral flow on each edge. 
\end{theorem}
\begin{proof}
Though the proof is well-known, e.g., ~\cite{Conitzer12}, we repeat it here for the sake of completeness. 
Consider induction on $n$, the number of nodes in the flow graph. 
The statement is vacuously true for $n=0$ and $n=1$, which serve as our base cases. 
Observe that we can write the linear program with $n-1$ constraints and thus there exists an optimal solution where at most $n-1$ edges have positive flow. 
By a simple averaging argument, there exists a vertex $v$ that has at most one incident edge $e$ with positive flow. 
Let $u$ be the other endpoint of the the edge $e=(u,v)$. 
Since it is the only edge incident to $v$, it must satisfy the entire demand of $v$. 
Because $v$ has integer demand, then $e$ has integer flow. 
However, the remainder of the graph has $n-1$ vertices and thus by induction, the remaining of the vertex demands are satisfied by a flow with integer demands. 
\end{proof}
We now adjust the integrality theorem to handle capacitated edges, thereby showing that the resulting solution for the min-cost flow problem in \figref{fig:flow} is integral.

\begin{corollary}
\corlab{cor:integral:sol}
Any minimum cost network flow problem with integral demands and capacities has an optimal solution with integral flow on each edge. 
\end{corollary}
\begin{proof}
The proof follows from a simple gadget to transform a min-cost flow problem with integer-capacitated edges into an uncapacitated min-cost flow problem. 
Suppose there exists a directed edge $e$ from $u$ to $v$ with capacity $c$ and cost $p$. 
Suppose furthermore that $u$ has demand $d_1$ and $v$ has demand $d_2$. 
Then we create an additional vertex $w$ and we replace $e$ with directed edges $e_1$ going from $u$ to $w$ and $e_2$ going from $v$ to $w$. 
We change the demand of $v$ to $d_2-c$, noting this can be negative. 
We also require vertex $w$ to have demand $c$. 
We then have cost $p$ on edge $e_2$ and cost $0$ on edge $e_2$. 
See \figref{fig:min:cost:transform} for an illustration of the transformation. 
Since the resulting graph after the reduction does not have any capacities on the edges, it follows from \thmref{thm:integral:sol} that there exists an integral solution to the original input problem. 
\end{proof}

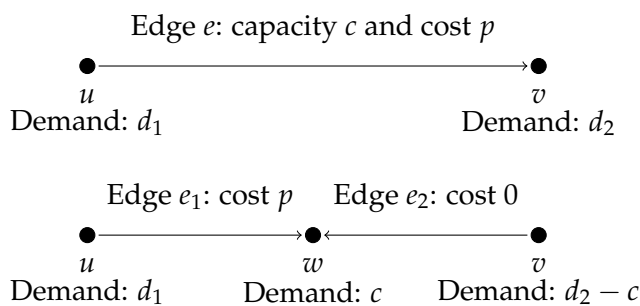
\begin{figure*}[!htb]
\centering
\begin{tikzpicture}[scale=0.5]
\filldraw (-6,4.5) circle (0.2);
\filldraw (6,4.5) circle (0.2);

\node at (-6,3.7){$u$};
\node at (-6,3){Demand: $d_1$};
\node at (6,3.7){$v$};
\node at (6,3){Demand: $d_2$};

\draw[->] (-5.7,4.5) -- (5.7,4.5);
\node at (0,5.5){Edge $e$: capacity $c$ and cost $p$};

\filldraw (-6,0) circle (0.2);
\filldraw (0,0) circle (0.2);
\filldraw (6,0) circle (0.2);

\draw[->] (-5.7,0) -- (-0.3,0);
\draw[<-] (0.3,0) -- (5.7,0);
\node at (-6,-0.8){$u$};
\node at (-6,-1.5){Demand: $d_1$};
\node at (-3,1){Edge $e_1$: cost $p$};
\node at (0,-0.8){$w$};
\node at (0,-1.5){Demand: $c$};
\node at (3,1){Edge $e_2$: cost $0$};
\node at (6,-0.8){$v$};
\node at (6,-1.5){Demand: $d_2-c$};
\end{tikzpicture}
\caption{Example of transformation of capacitated min-cost flow problem into uncapacitated min-cost flow problem.}
\figlab{fig:min:cost:transform}
\end{figure*}

Hence, the min-cost flow solution defines a valid clustering that approximately optimal with respect to the $\ell_2^2$ min-sum $k$-clustering objective. 
However, we further want to show the property holds for the solution returned by a linear program solver. 
In fact, it is well-known the constraint matrix is totally unimodular, i.e., all submatrices have determinant $-1$, $0$, or $1$. 
\begin{theorem}
[Theorem 19.1 in \cite{schrijver1998theory}]
\thmlab{thm:unimodular:corner}
Let $A$ be a totally unimodular matrix and let $b$ be an integer vector. 
Then all vertices of the polyhedron $P=\{x\,\mid\,Ax\le b\}$ are integral. 
\end{theorem}
Since the solution of a linear program must lie at a vertex of the feasible polytope, then \thmref{thm:unimodular:corner} implies any solution to the linear program will also be integral. 
Thus a valid clustering can be recovered by using the output of a linear program solver. 
We recall the following various implementations of solvers for linear programs.
\begin{theorem}
\thmlab{thm:lp:solver}
\cite{Karmarkar84,Vaidya89a,Vaidya90,LeeS15,LeeSZ19,CohenLS21,JiangSWZ21}
There exists an algorithm that solves a linear program with $n$ variables that can be encoded in $L$ bits, using $\poly(n,L)$ time.
\end{theorem}
Putting things together, we have the following guarantees for our learning-augmented algorithm.
\begin{theorem}
There exists a polynomial-time algorithm that uses a label predictor with error rate $\alpha$ and outputs a $\frac{1+\gamma_\alpha\alpha}{(1-\alpha)^2}$-approximation to min-sum $k$-clustering, where $\gamma_\alpha$ is the fixed constant from \thmref{thm:each:cluster}. 
\end{theorem}
\begin{proof}
Correctness follows from \corref{cor:aug:correctness}. 

For the runtime analysis, first observe that the centers $c_1,\ldots,c_k$ can be computed in polynomial time by \thmref{thm:each:cluster}. 
Subsequently, $\calF$ can be written as a linear programming problem with at most $\poly(n)$ constraints and variables. 
Therefore, the desired claim follows by running any polynomial-time linear programming solver, i.e., \thmref{thm:lp:solver} and observing that the output solution induces a valid clustering, by \thmref{thm:unimodular:corner}. 
\end{proof}

\section*{Acknowledgements}
The work was conceptualized while all the authors were visiting the Institute for Emerging CORE Methods in Data Science (EnCORE) supported by the NSF grant 2217058. 
Karthik C.\ S.\ was supported by the National Science Foundation under Grants CCF-2313372 and CCF-2443697, a grant from the Simons Foundation, Grant Number 825876, Awardee Thu D. Nguyen, and partially funded by the Ministry of Education and Science of Bulgaria's support for INSAIT, Sofia University ``St. Kliment Ohridski'' as part of the Bulgarian National Roadmap for Research Infrastructure. Euiwoong Lee was supported in part by NSF grant CCF-2236669 and Google. Yuval Rabani was supported in part by ISF grants 3565-21 and 389-22, and by BSF grant 2023607. Chris Schwiegelshohn was partially supported by the Independent Research Fund
Denmark (DFF) under a Sapere Aude Research Leader grant No 1051-00106B.
Samson Zhou is supported in part by NSF CCF-2335411. 
The work was conducted in part while Samson Zhou was visiting the Simons Institute for the Theory of Computing as part of the Sublinear Algorithms program.

\def\shortbib{0}
\bibliographystyle{alpha}
\bibliography{references}

\end{document}